\documentclass[conference]{IEEEtran}
\IEEEoverridecommandlockouts
\usepackage{cite}
\usepackage{subcaption}

\usepackage{dblfloatfix}
\usepackage{overpic}
\usepackage{amsmath,amssymb,amsfonts}
\usepackage{graphicx}
\usepackage{textcomp}
\usepackage{xcolor}
\usepackage{float}
\usepackage{amsthm}
\usepackage{graphicx}
\usepackage{epstopdf}
\usepackage{amsmath,bm}
\usepackage{amsfonts}
\usepackage{amssymb}
\usepackage{color}
\usepackage{multirow}
\usepackage{multicol}
\usepackage{soul,xcolor}
\usepackage{tcolorbox}
\usepackage{algorithm}
\usepackage{algpseudocode}
\usepackage{empheq}

\theoremstyle{plain}
\newtheorem{theorem}{Theorem}

\newtheorem{remark}{Remark}

\newcommand{\vect}[1]{\mathbf{#1}}

\def\diag{\mathrm{diag}}

\def\tr{\mathrm{tr}}

\def\Htran{\mbox{\tiny $\mathrm{H}$}}
\def\Ttran{\mbox{\tiny $\mathrm{T}$}}
\def\CN{\mathcal{N}_{\mathbb{C}}} 
\def\imagunit{\mathsf{j}} 
\def\m{\rm}

\begin{document}

\title{NCR vs. Passive/Active RIS: How Much NCR Amplification is Required to Beat RIS?}

\author{\"Ozlem Tu\u{g}fe Demir, Ozan Alp Topal, Cicek Cavdar, and Emil Bj{\"o}rnson \vspace{-2cm}
\thanks{ \"O. T. Demir is with the Department of Electrical and Electronics Engineering, Bilkent University, Ankara, Turkiye (ozlemtugfedemir@bilkent.edu.tr).  O. A. Topal, C. Cavdar, and E.~Bj\"ornson are with the Department of Communication Systems, KTH Royal Institute of Technology, Stockholm, Sweden (oatopal@kth.se, cavdar@kth.se, emilbjo@kth.se). 

The work by \"O. T. Demir was supported by 2232-B International Fellowship for Early Stage Researchers Programme funded by the Scientific and Technological Research Council of T\"urkiye. \newline 
}
}

\maketitle
\begin{abstract}
This paper investigates the fundamental tradeoff between reconfigurable intelligent surfaces (RISs) and network-controlled repeaters (NCRs) in terms of achievable signal-to-noise ratio (SNR). Considering an uplink system with a multi-antenna base station (BS) and a single-antenna user equipment (UE), we derive closed-form SNR expressions for passive RIS-, active RIS-, and NCR-assisted communication under line-of-sight propagation between the BS-RIS/NCR and RIS/NCR-UE. Both narrowband and wideband transmissions are analyzed, with and without the presence of a direct BS--UE link. Our analysis reveals a key structural difference: while the SNR achieved with RISs grows unboundedly with the number of RIS elements, the SNR provided by an NCR is fundamentally limited by the UE--repeater channel due to noise amplification. Nevertheless, we show that NCRs can outperform both passive and active RISs when deployed close to the UE, provided that sufficient amplification is available. Numerical results based on realistic path loss models quantify the amplification levels required for NCRs to outperform RISs across different deployment geometries and system dimensions. These findings provide clear design guidelines for the practical integration of RISs and NCRs in future wireless networks.
\end{abstract}


\section{Introduction}

Reconfigurable intelligent surfaces (RISs) are nearly passive electromagnetic structures that can reconfigure the propagation environment between a transmitter and a receiver~\cite{liu2021reconfigurable}. They are particularly attractive for establishing virtual line-of-sight (LOS) links when the direct path is blocked. Even in the presence of a direct link, RISs can enhance spectral efficiency by inducing phase shifts that align reflected signals for constructive superposition. Although RISs are not the first technology proposed to manipulate the radio propagation environment, they provide a low-cost, real-time, full-duplex relaying solution with very low radiated power consumption~\cite{bjornson2021reconfigurable}. This hardware simplicity, however, comes at the cost of requiring very large apertures to compensate for the severe double-cascaded pathloss inherent to RIS-assisted links.

To mitigate this limitation, RIS elements can be equipped with active circuits that provide signal amplification in addition to phase shifting. In contrast, a \emph{passive RIS} is limited to applying unit-modulus phase shifts without any signal amplification. The so-called \emph{active RIS}, has been shown to outperform passive RISs in several scenarios~\cite{zhang2022active,zhi2022active,basar2024reconfigurable}. Nevertheless, since active RISs also generate and amplify receiver noise at the RIS elements, the amplification must be carefully designed to avoid excessive noise levels. For both passive and active RISs, the primary SNR gain stems from the large number of RIS elements, which in turn requires substantial control signaling to configure the RIS phase or amplitude responses in every coherence block.

An alternative low-cost, real-time relaying technology is the \emph{network-controlled repeater} (NCR), which has already been standardized in 3GPP Release~18~\cite{wen2024shaping} and is a well-established solution for coverage extension in wireless networks~\cite{tsai2010capacity}. NCRs amplify and forward the received signal—along with the associated receiver noise— with processing delays on the order of a few hundred nanoseconds~\cite{willhammar2025achieving}, effectively acting as active scatterers. Compared to active RIS architectures, NCRs can typically support higher amplification levels due to their centralized hardware design, whereas the per-element amplification in active RIS is limited by circuit constraints. As will be shown in this paper, this difference can enable NCRs to outperform both passive and active RISs in certain operating regimes.  

Beyond coverage enhancement, NCRs have recently been proposed as a means to increase macro diversity within a cell, giving rise to the concept of repeater-assisted massive MIMO (RA-MIMO)~\cite{willhammar2025achieving,topal2025fair}. In addition, NCR operation under practical hardware effects and wideband signaling has recently been investigated in dedicated repeater-centric studies~\cite{demir2026pa,evgur2026ofdm}.

Both RISs and NCRs operate without requiring full baseband processing and decoding at intermediate nodes. Nevertheless, RIS-based solutions suffer from several practical limitations, including bulky implementations and significant control and training overhead~\cite{willhammar2025achieving}. These limitations become particularly critical in emerging upper midband and mmWave 6G deployments, where coverage remains a major bottleneck due to increased path loss and blockage sensitivity. In such regimes, cost-effective and low-latency coverage extension solutions are essential.

Prior works have investigated RISs and NCRs as coverage extenders mainly through numerical and system-level comparisons. For instance, the study in~\cite{ayoubi2022network} shows that the relative performance of RISs and NCRs strongly depends on the deployment geometry, suggesting that different scenarios may favor different technologies. Similarly,~\cite{guo2022comparison} demonstrates through simulations that NCR-assisted networks can often outperform RIS-assisted systems in terms of achievable user equipment (UE) rates under certain parameter settings. While these studies provide valuable insights, they offer limited analytical understanding of the fundamental conditions under which one technology outperforms the other.

A fundamental performance distinction among passive RISs, active RISs, and NCRs arises from their noise behavior. Passive RISs are effectively noise-free, whereas both active RISs and NCRs introduce and amplify receiver noise. However, the extent and impact of this noise amplification differ due to their underlying architectures and achievable amplification levels, which makes system-level modeling essential to accurately capture the tradeoffs among these technologies~\cite{sun2023performance}.

To address these limitations, this paper develops an analytical framework to compare RIS- and NCR-assisted systems. In particular, we derive closed-form conditions under which NCRs outperform passive and active RISs in terms of uplink signal-to-noise ratio (SNR) in a system with a multi-antenna base station (BS) and a single-antenna UE. To obtain interpretable insights, we focus on LOS propagation over the BS--RIS/NCR and RIS/NCR--UE links and consider both narrowband and wideband transmissions, with and without a direct BS--UE component. The resulting expressions provide explicit design guidelines that quantify the interplay among the number of BS antennas, the number of RIS elements, the repeater amplification gain, and the deployment geometry. The main contributions of this paper have been outlined as follows:
\begin{itemize}
    \item 
    In contrast to prior NCR literature, which primarily focuses on repeater-aided MIMO performance \cite{topal2025fair,andersson2025repeater}, or on modeling repeater-assisted sensing/communication tradeoffs \cite{chowdhury2025performance}, this paper derives tractable closed-form SNR expressions and uses them to obtain explicit inequalities that determine when NCRs beat RIS-based links.

    \item We characterize geometry-dependent operating regimes and reveal how deployment conditions, amplification capability, and array sizes determine the preferable technology, thereby providing guidelines for hybrid RIS–NCR deployments.

    \item We extend the analysis to both narrowband and wideband systems and show that the key qualitative insights remain consistent across different transmission regimes.

\end{itemize}

The remainder of this paper is organized as follows. Section~II introduces the system and channel models for passive RIS-, active RIS-, and NCR-assisted uplink communication. In Section~III, we derive closed-form SNR expressions and establish analytical conditions under which NCRs outperform passive and active RISs in the absence of a direct BS--UE link. Section~IV extends the analysis to scenarios with a direct non-LOS (NLOS) BS--UE path and examines the resulting performance interplay. In Section~V, we generalize the comparison to wideband transmissions and discuss the impact of frequency-dependent phase variations. Numerical results are presented throughout to illustrate the analytical findings under realistic deployment geometries. Finally, Section~VI concludes the paper and outlines directions for future work.

\section{System and Channel Model} \label{sec2}

We consider an uplink communication scenario between a BS equipped with $M$ antennas and a single-antenna UE. In this section, we focus on narrowband transmission and later extend the model to the wideband case. The direct BS--UE link is assumed to be completely blocked; its impact is analyzed separately in a later section. Hence, communication is assisted either by a passive RIS, an active RIS, or an NCR.

We first derive the SNR for passive and active RIS-assisted communication, and then derive the corresponding expression for the NCR-assisted case. Throughout the analysis, we do not impose an explicit total power budget constraint. Instead, all technologies are studied under maximum transmit power constraints imposed by circuit-level limitations.

This modeling choice is motivated by the fact that, in both active RIS and NCR architectures, the radiated transmit power is often negligible compared to the static power consumption of the hardware. In practice, the overall power consumption is typically dominated by technology-dependent static components, such as FPGA-based control boards in RISs and power amplifier circuitry in NCRs, while the transmit power contributes only marginally to the total energy consumption. 

For instance, reported implementations in~\cite{wang2024reconfigurable} indicate static power consumptions on the order of $6.52$\,W and $15.73$\,W for two different RIS deployments, whereas NCR prototypes have been reported to consume around $10$\,W in~\cite{willhammar2025achieving}. These values should be interpreted as indicative examples rather than definitive benchmarks, since both RIS and NCR technologies are still under active development and their hardware characteristics can vary significantly across implementations.

Moreover, for practically relevant numbers of RIS elements, the static power consumption of passive RISs, active RISs, and NCR-based solutions is typically of the same order of magnitude. As a result, differences in total power consumption across these technologies are largely driven by static hardware components and do not scale significantly with the transmitted signal power.

\subsection{Passive RIS-assisted communication}

The passive RIS consists of $N$ elements, and the channel between the UE and the RIS is assumed to be free-space LOS, expressed as  
\begin{align}
    \vect{h}_1 = \sqrt{\beta_1} \vect{a}_1,
\end{align}
where $\vect{a}_1$ denotes the array response vector with unit-modulus entries.\footnote{The assumption of identical channel gains across all antennas remains valid even in the radiative near-field, provided that the UE--RIS distance is at least twice the aperture length of the RIS \cite{bjornson2020power}.}

The RIS--BS channel is modeled as a far-field free-space LOS channel, given by  
\begin{align}
    \vect{H}_2 = \sqrt{\beta_2} \vect{a}_{2,1} \vect{a}_{2,2}^{\Ttran},
\end{align}
where $\beta_2$ is the channel gain and $\vect{a}_{2,1}$ and $\vect{a}_{2,2}$ are the corresponding array response vectors with unit-modulus entries. Then, the overall RIS-assisted channel becomes  
\begin{align}
    \vect{h} = \vect{H}_2 \vect{\Psi}_{\rm P} \vect{h}_1,
\end{align}
where the matrix $\vect{\Psi}_{\rm P}=\diag(\psi_{{\rm P},1},\ldots, \psi_{{\rm P},N})$ includes the phase-shift responses of the passive RIS at its diagonal entries, which are complex scalars on the unit circle. The received signal at the BS is then given by
\begin{align}
    \vect{y} = \sqrt{P}\vect{h}s + \vect{n},
\end{align}
where $s$ is the transmitted symbol, $P > 0$ is the UE transmit power, and $\vect{n}\sim\CN(\vect{0},\sigma^2\vect{I}_M)$ denotes the additive white Gaussian noise (AWGN) at the BS receiver. Under optimal maximum-ratio (MR) receive combining, the resulting SNR is
\begin{align}
    \mathrm{SNR}_{\rm P-RIS} = \frac{P \Vert \vect{h} \Vert^2}{\sigma^2}.
\end{align}
 The RIS phase-shift configuration that maximizes the SNR is  
\begin{align}
    \psi_{{\rm P},n} = \exp\left(\imagunit(-\angle{[\vect{a}_1]_n} - \angle{[\vect{a}_{2,2}]_n})\right),
\end{align}
which yields  
\begin{align}
    \vect{h} = \sqrt{\beta_1 \beta_2} N \vect{a}_{2,1}.
\end{align}
The resulting maximum SNR obtained with a passive RIS is  
\begin{align}
    \mathrm{SNR}_{\rm P-RIS} = \frac{P \beta_1 \beta_2 N^2 M}{\sigma^2}.
\end{align} 

\subsection{Active RIS-assisted communication}

The active RIS is also equipped with $N$ elements, and the channels between the UE and the RIS, as well as between the RIS and the BS, are assumed to be identical to those defined in the previous section. The received signal at the active RIS is given in vector form as
\begin{align}
    \vect{y}_1 = \sqrt{P}\vect{h}_1s+\vect{n}_1
\end{align}
where $\vect{n}_1\sim \CN(\vect{0}, \sigma^2\vect{I}_N)$ denotes the independent receiver noise at the active RIS elements. The amplified signal by the active RIS is given as
\begin{align}
    \tilde{\vect{y}}_1 = \vect{\Psi}_{\rm A} \vect{y}_1
\end{align}
where the $\vect{\Psi}_{\rm A}=\diag(\psi_{{\rm A},1},\ldots, \psi_{{\rm A},N})$ matrix includes the complex responses of the active RIS at its diagonal entries. The magnitudes of these complex responses may exceed unity in an active RIS; however, this inevitably leads to noise amplification. In practice, the maximum amplification factor is constrained by the hardware limitations of the RIS element circuitry \cite{gavriilidis2025active}.

Then, the received signal at the BS is given as 
\begin{align}
    \vect{y}_2 = \vect{H}_2 \tilde{\vect{y}}_1 +\vect{n}_2 = \sqrt{P}\vect{H}_2\vect{\Psi}_{\rm A} \vect{h}_1s+\vect{H}_2\vect{\Psi}_{\rm A}\vect{n}_1 + \vect{n}_2
\end{align}
where $\vect{n}_2\sim\CN(\vect{0}, \sigma^2\vect{I}_M)$ is the AWGN at the BS receiver. Since both the desired signal and the amplified noise arrive through the same spatial direction $\vect{a}_{2,1}$, the SNR is maximized by using the receive combining vector $\vect{a}_{2,1}$. The resulting SNR is
\begin{align}
\mathrm{SNR}_{\rm A-RIS}=\frac{P\beta_1\beta_2M\left|\vect{a}_{2,2}^{\Ttran}\vect{\Psi}_{\rm A}\vect{a}_1\right|^2}{\sigma^2\left(1+\beta_2M\sum_{n=1}^N\left|\psi_{{\rm A},n}\right|^2\right)}.
\end{align}
The RIS configuration that maximizes the SNR under a maximum amplification gain of $\alpha_{\rm A-RIS}$ per element can be obtained as follows. The phase shifts only affect the coherent combination of the desired signal and are therefore chosen as in the passive RIS case to achieve phase alignment. Moreover, since the SNR is monotonically increasing in each amplification coefficient $|\psi_{{\rm A},n}|$, the maximum gain is attained by setting all amplitudes to their maximum value $\alpha_{\rm A-RIS}$. This yields
\begin{align}
    \psi_{{\rm A},n} =\alpha_{\rm A-RIS} \exp\left(\imagunit(-\angle{[\vect{a}_1]_n} - \angle{[\vect{a}_{2,2}]_n})\right),
\end{align}
which leads to the maximum SNR  
\begin{align}
    \mathrm{SNR}_{\rm A-RIS} = \frac{\alpha_{\rm A-RIS}^2P \beta_1 \beta_2 N^2 M}{\sigma^2\left(1+\alpha_{\rm A-RIS}^2\beta_2MN\right)}.
\end{align}

\subsection{NCR-assisted communication}

Unlike active RIS architectures, which consist of many distributed elements each applying small amplification factors, an NCR is a centralized device that applies a single high-gain amplification stage to the received signal. Although both systems can be described within a similar signal model, their hardware architectures and operating regimes are fundamentally different.

We now consider NCR-assisted communication, where the NCR applies an amplification gain $\alpha_{\rm NCR}$. The scalar channel between the UE and the NCR is modeled as  
\begin{align}
    h_1 = \sqrt{\beta_1} e^{\imagunit \varphi_1},
\end{align}
while the NCR--BS channel is  
\begin{align}
    \vect{h}_2 = \sqrt{\beta_2} \vect{a}_{2,1}.
\end{align}

The received signal at the repeater is  
\begin{align}
    y_1 = \sqrt{P} h_1 s + n_1,
\end{align}
where $n_1 \sim \CN(0, \sigma^2)$ denotes independent receiver noise at repeater. The BS then receives  
\begin{align}
    y_2 &= \alpha_{\rm NCR} y_1 \vect{h}_2 + \vect{n}_2 \nonumber \\
        &= \alpha_{\rm NCR} \sqrt{P \beta_1 \beta_2}\, \vect{a}_{2,1} e^{\imagunit \varphi_1} s 
        + \alpha_{\rm NCR} \sqrt{\beta_2}\, n_1 \vect{a}_{2,1} + \vect{n}_2,
\end{align}
with $\vect{n}_2 \sim \CN(\vect{0}, \sigma^2 \vect{I}_M)$. The optimal linear receiver is MR combining since both the desired signal and noise share the same spatial direction. The maximized SNR is then  
\begin{align}
    \mathrm{SNR}_{\rm NCR} = \frac{\alpha_{\rm NCR}^2 P M \beta_1 \beta_2}{\sigma^2 \left( 1 + \alpha_{\rm NCR}^2 \beta_2 M \right)}.
\end{align}

\section{Comparison between passive/active RIS and NCR}\label{sec3}

In this section, we will compare the SNRs achieved with passive/active RIS and NCR, respectively, and quantify the amplification gain required by NCR to beat RISs. 

\subsection{Comparison between passive RIS and NCR}
For the NCR to outperform the passive RIS, $\alpha_{\rm NCR}$ must satisfy  
\begin{align}
    &\frac{\alpha_{\m NCR}^2 P M \beta_1 \beta_2}{\sigma^2 (1 + \alpha_{\m NCR}^2 \beta_2 M)} 
    \geq \frac{P \beta_1 \beta_2 N^2 M}{\sigma^2} \\
    &\Longleftrightarrow\alpha_{\m NCR}^2 \left( 1 - N^2 M \beta_2 \right) \geq N^2 \\
    &\Longleftrightarrow\alpha_{\m NCR} \geq \frac{N}{\sqrt{1 - N^2 M \beta_2}}, \label{eq:condition-no-direct}
\end{align}
under the assumption $1 - N^2 M \beta_2 > 0$. By contrast, if $N^2 M \beta_2 \geq 1$, the RIS always achieves a higher SNR. For realistic $\beta_2$ values, however, this condition requires extremely large numbers of RIS elements and/or BS antennas. For example, even if $\beta_2$ is relatively large (e.g., $\beta_2 = -60$\,dB), a BS with $M = 100$ antennas would require at least $N \geq 100$ RIS elements.  

Importantly, the RIS-assisted SNR has no inherent upper bound (as a function of $N$), whereas the repeater-assisted SNR is upper-bounded as
\begin{align}
\mathrm{SNR}_{\rm NCR} 
= \frac{\alpha_{\rm NCR}^2 P M \beta_1 \beta_2}{\sigma^2 (1 + \alpha_{\rm NCR}^2 \beta_2 M)}
\leq \frac{P \beta_1}{\sigma^2}, \label{eq:SNR-limit}
\end{align}
where the upper bound follows from taking $\alpha_{\rm NCR} \to \infty$. The upper bound
corresponds to the SNR of the UE--NCR link if the repeater were replaced by a co-located receiver with identical noise statistics. When $N^2 M \beta_2 \geq 1$, it is evident that $\mathrm{SNR}_{\rm P-RIS}$ always exceeds this upper bound.  

The lower bound on $\alpha_{\m NCR}$ in \eqref{eq:condition-no-direct}, i.e., $\frac{N}{\sqrt{1 - N^2 M \beta_2}}$, is independent of $\beta_1$ but increases with $N$, $M$, and $\beta_2$. Thus, larger system dimensions and stronger BS links demand increasingly higher repeater amplification gains to outperform the passive RIS.

\subsubsection*{Numerical illustration}

We set the wavelength to $\lambda = 0.02$\,m, corresponding to $f_c = 15$\,GHz, which lies in the upper midband. We adopt the free-space path loss model with isotropic antennas, i.e., $\beta_2 = \lambda^2 / (4\pi d_2)^2$, where $d_2$ denotes the propagation distance. We vary $d_2$ in Fig.~\ref{fig1} for different values of $N$ and show the required value of $\alpha_{\rm NCR}$ to provide the largest SNR, obtained from \eqref{eq:condition-no-direct}.

Fig.~\ref{fig1} shows that, when the RIS is moderately sized, the NCR outperforms the RIS after only a few meters (note that $\alpha_{\rm NCR}$ is omitted in cases where the RIS always performs better). However, for larger RIS sizes, if the BS--RIS/NCR distance is short, the RIS consistently outperforms the NCR. Conversely, when $N$ is small, the NCR requires only moderate amplification to outperform the RIS across all distances. The literature reports that repeater amplification gains of up to $90/2=45$\,dB are achievable \cite{bai2025repeater}, the amplification levels required in Fig.~\ref{fig1} are well below this limit, which justifies the use of the term \emph{moderate}.

Another important observation is that, as the distance $d_2$ increases, the channel gain $\beta_2$ decreases toward zero. Consequently, the required amplification factor $\alpha_{\rm NCR}$ converges to $N$, as predicted by \eqref{eq:condition-no-direct}. 

\begin{figure}[t!]
        \centering
	\begin{overpic}[width=0.98\columnwidth,trim=0.2cm 0cm 0.5cm 0.5cm,clip,tics=10]{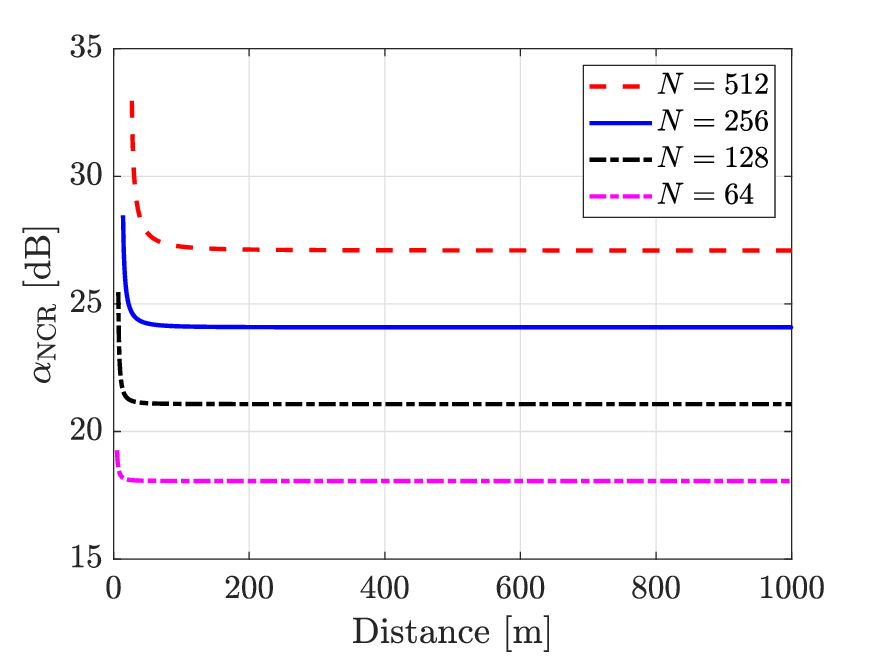}
\end{overpic} 
        \caption{The required amplification gain $\alpha_{\rm NCR}$ of the repeater to outperform the passive RIS in terms of the distance between the NCR/RIS and BS. }
        \label{fig1}
\end{figure}

In Fig.~\ref{fig2}, we illustrate the SNR achieved by an NCR and a passive RIS as a function of the distance $d_2$, for different values of the amplification factor $\alpha_{\rm NCR}$ and the number of RIS elements $N$. The system parameters are set to $P=20$\,dBm, $\beta_1=-87$\,dB (which corresponds to $d_1\approx 35.6$\,m in free-space path loss formula), and $\sigma^2=-117$\,dBm, corresponding to a bandwidth of $100$\,MHz and a noise figure of $7$\,dB, which yields $P\beta_1/\sigma^2=20$\,dB. The BS is equipped with $M=1024$ antennas.

When the NCR amplification is higher ($\alpha_{\rm NCR}=45$\,dB), the SNR closely approaches the ultimate limit $P\beta_1/\sigma^2$ in \eqref{eq:SNR-limit} over the entire range of considered distances. In contrast, when $\alpha_{\rm NCR}=40$\,dB, the gap to this ultimate limit increases with $d_2$. While a passive RIS can provide performance gains when placed very close to the BS, the NCR with moderate amplification outperforms a passive RIS with $512$ elements beyond a certain distance. As the distance increases, the performance gap between NCR and passive RIS becomes significant.

Importantly, under these system parameter assumptions, the transmit power of the NCR with $\alpha_{\rm NCR}=40$\,dB remains extremely low, amounting to only $-46.96$\,dBm. The transmit power is around $-36.96$\,dBm for $\alpha_{\rm NCR}=45$\,dB, which is still very low.

\begin{figure}[t!]
        \centering
	\begin{overpic}[width=0.98\columnwidth,trim=0.2cm 0cm 0.5cm 0.5cm,clip,tics=10]{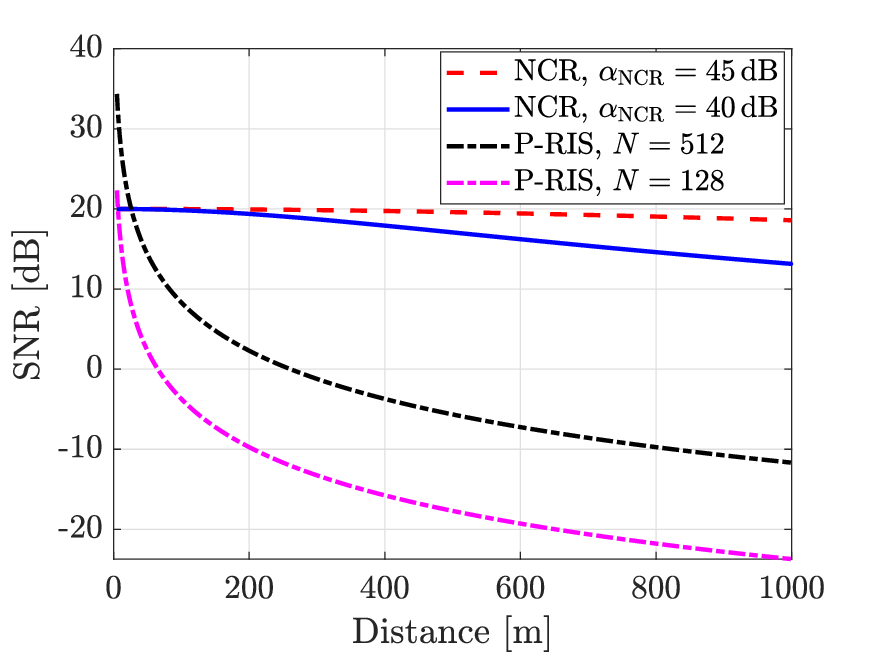}
\end{overpic} 
        \caption{The SNR obtained with passive RIS and NCR in terms of the distance between the NCR/RIS and BS. }
        \label{fig2}
\end{figure}

\subsection{Comparison between active RIS and NCR}

For the NCR to outperform the active RIS, $\alpha_{\rm NCR}$ must satisfy  
\begin{align}
    &\frac{\alpha_{\m NCR}^2 P M \beta_1 \beta_2}{\sigma^2 (1 + \alpha_{\m NCR}^2 \beta_2 M)} 
    \geq \frac{\alpha_{\rm A-RIS}^2P \beta_1 \beta_2 N^2 M}{\sigma^2\left(1+\alpha_{\rm A-RIS}^2\beta_2MN\right)}, \label{eq:condition-no-direct-active}
\end{align}
which is equivalent to 
\begin{align}
  &  \alpha_{\m NCR} \geq \sqrt{\frac{\frac{\alpha_{\m A-RIS}^2N^2}{1+\alpha_{\m A-RIS}^2\beta_2MN}}{1-\beta_2M\frac{\alpha_{\m A-RIS}^2N^2}{1+\alpha_{\m A-RIS}^2\beta_2MN}}} \\
    & \Longleftrightarrow \alpha_{\m NCR} \geq \frac{\alpha_{\rm A-RIS}N}{\sqrt{1-(N^2-N)\alpha_{\rm A-RIS}^2M\beta_2}} \label{eq:condition-no-direct-active2}
\end{align}
under the assumption that $1-(N^2-N)\alpha_{\rm A\text{-}RIS}^2 M \beta_2>0$. 
By contrast, if $(N^2-N)\alpha_{\rm A\text{-}RIS}^2 M \beta_2 \geq 1$, the active RIS always achieves a higher SNR than the NCR. 
Compared to the corresponding condition in the passive RIS case, namely $1-N^2 M \beta_2>0$, this assumption is more restrictive, implying that outperforming an active RIS becomes more challenging for the NCR.

Nevertheless, if $\beta_2$ is sufficiently small, the NCR can still outperform the active RIS. This behavior can be attributed to the fact that NCRs are expected to support higher amplification levels compared to active RIS architectures due to their centralized hardware design. Consequently, in certain operating regimes, this increased amplification capability may compensate for the noise amplification and enable NCRs to outperform active RISs.

For example, prior works report that the maximum amplification gain of a repeater can reach up to $45$\,dB, whereas the per-element amplification gain of an active RIS is typically more limited, e.g., $\alpha_{\rm A-RIS}=30$, corresponding to $14.77$\,dB in the circuit model considered in~\cite{gavriilidis2025active}. These values should, however, be interpreted as indicative rather than definitive, since both technologies are still under active development.

Interestingly, as $\beta_2 \to 0$, the required amplification factor for the NCR converges to $\alpha_{\rm A\text{-}RIS} N$, which results in an $N$-fold increase in radiated power from the NCR compared to the active RIS. 
This can be seen by comparing $\alpha_{\rm NCR}^2=\alpha_{\rm A\text{-}RIS}^2 N^2$ for the NCR with $\alpha_{\rm A\text{-}RIS}^2 N$ for the active RIS.  
However, given the extremely low transmit power levels required for the NCR in the considered regime, this increase does not constitute a practical bottleneck. In contrast, when the RIS employs an excessive number of elements, the power consumption associated with RIS
controlling and configuring becomes dominant and can reach tens of watts, as reported in \cite{wang2024reconfigurable}.

\subsubsection*{Numerical illustration}
To quantify the required NCR amplification to beat active RIS, we consider the previously defined simulation setup, where the wavelength is $\lambda = 0.02$\,m, corresponding to $f_c = 15$\,GHz, and adopt the free-space path loss model with $\beta_2 = \lambda^2 / (4\pi d_2)^2$, where $d_2$ denotes the propagation distance. We fix $M = 1024$ and vary $d_2$ in Fig.~\ref{fig3}, where we show the required NCR amplification for different values of $N$.  

Fig.~\ref{fig3} shows that, when the RIS is moderately sized and the BS is located sufficiently far away from the NCR/RIS, the NCR outperforms the RIS  with moderate values of amplification. On the other hand, when the number of RIS elements is high, e.g., $N=512$, it is only possible for the NCR to outperform active RIS when the propagation distance is very long. Similar to the passive RIS, as the distance $d_2$ increases, the channel gain $\beta_2$ decreases toward zero and the required amplification factor $\alpha_{\rm NCR}$ converges to $\alpha_{\rm A-RIS}N$, as predicted by \eqref{eq:condition-no-direct-active2}.

\begin{figure}[t!]
        \centering
	\begin{overpic}[width=0.98\columnwidth,trim=0.2cm 0cm 0.5cm 0.5cm,clip,tics=10]{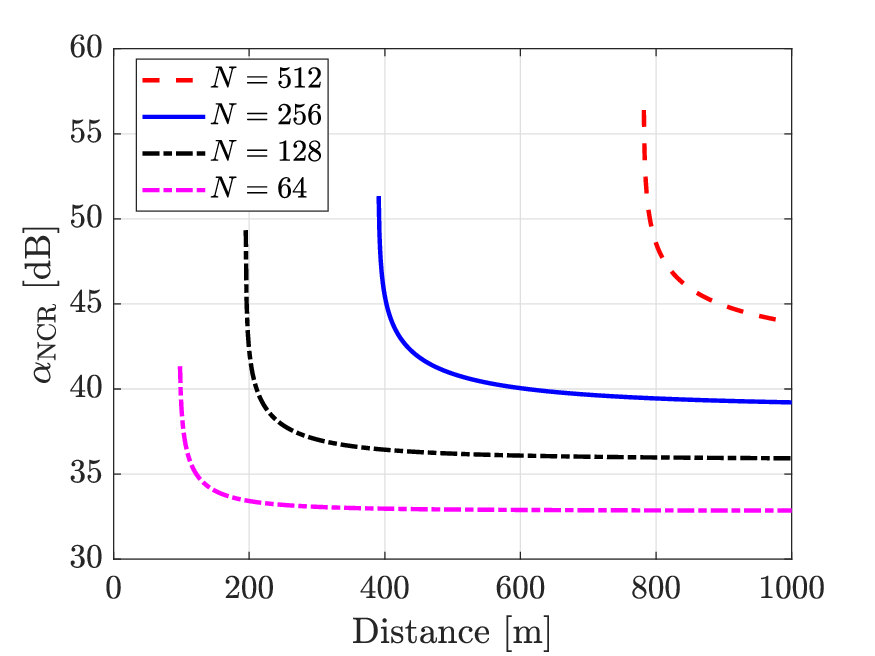}
\end{overpic} 
        \caption{The required amplification gain $\alpha_{\rm NCR}$ of the repeater to outperform the active RIS in terms of the distance between the repeater/RIS and BS. }
        \label{fig3}
\end{figure}

In Fig.~\ref{fig4}, we illustrate the SNR achieved by an NCR and an active RIS as a function of the distance $d_2$, for different values of the amplification factor $\alpha_{\rm NCR}$ and the number of RIS elements $N$. The system parameters are set to $P=20$\,dBm, $\beta_1=-87$\,dB, and $\sigma^2=-117$\,dBm, corresponding to a bandwidth of $100$\,MHz and a noise figure of $7$\,dB, which yields $P\beta_1/\sigma^2=20$\,dB. The BS is equipped with $M=1024$ antennas, and the amplification gain per element of active RIS is $\alpha_{\rm A-RIS}=30$.

Compared with Fig.~\ref{fig2}, higher SNR values are obtained with the active RIS. Nevertheless, as long as the active RIS does not have an excessive number of RIS elements, the NCR outperforms the RIS at large distances. However, when the RIS is positioned close to the BS, using an RIS is advantageous over a NCR. Here, one should note that the distance from the UE to the RIS/NCR stays the same in this simulation setup.

To make a comparison in which the UE-to-BS distance remains constant while the positions of the RIS and NCR vary, we next consider the following geometry.
Specifically, we consider a three-dimensional setup in which the BS, equipped with $M=1024$ antennas, is located at $(0,0,10)$\,m. The UE is positioned at $(1000,0,0)$\,m, while the RIS or NCR is placed at $(r_x,10,10)$\,m. We evaluate the channel gains as functions of $r_x$, which denotes the horizontal position of the RIS/NCR. The active RIS employs an amplification factor of $\alpha_{\rm A-RIS}=30$. For the NCR, both a maximum amplification constraint of $45$\,dB and a total transmit power constraint of $1$\,W are imposed.

In Fig.~\ref{fig5}, we illustrate the SNR achieved by the NCR, passive RIS, and active RIS as a function of their deployment location, while keeping the BS and UE positions fixed. Due to the double-cascaded path loss, the SNR attained with the passive RIS remains below $0$\,dB for all considered locations. In contrast, the active RIS provides substantially higher SNR values, especially when placed in close proximity to the UE. For the active RIS configuration with $N=512$ elements, the NCR outperforms the active RIS when it is positioned closer to the UE. On the other hand, when $N=128$, the NCR achieves markedly higher SNR than active RIS for most deployment locations. Since both the NCR and the active RIS amplify not only the desired signal but also the noise, their SNR-optimal placement is closer to the UE.

\begin{figure}[t!]
        \centering
	\begin{overpic}[width=0.98\columnwidth,trim=0.2cm 0cm 0.5cm 0.5cm,clip,tics=10]{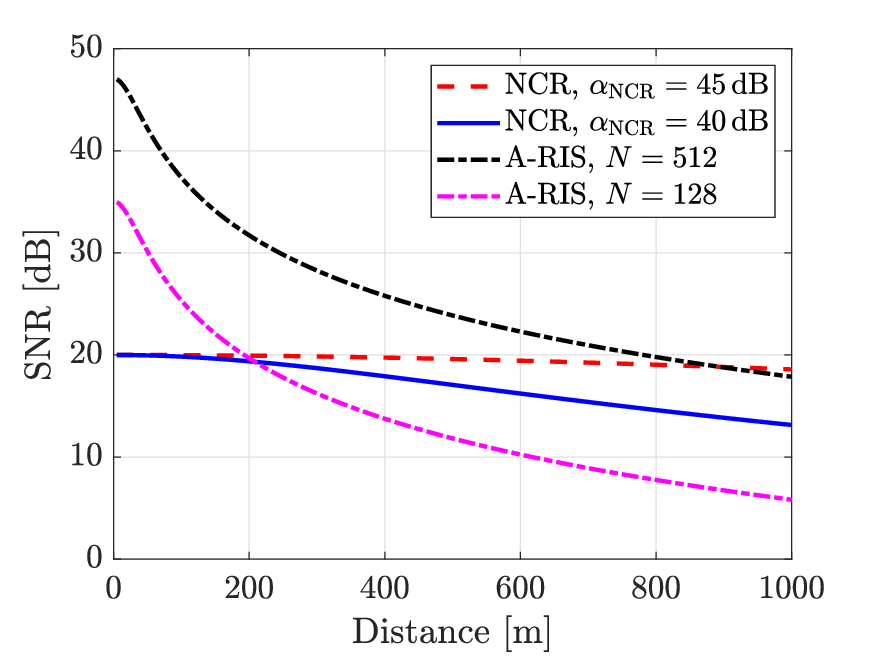}
\end{overpic} 
        \caption{The SNR obtained with active RIS and NCR in terms of the distance between the NCR/RIS and BS. }
        \label{fig4}
\end{figure}

\begin{figure}[t!]
        \centering
	\begin{overpic}[width=0.98\columnwidth,trim=0.2cm 0cm 0.5cm 0.5cm,clip,tics=10]{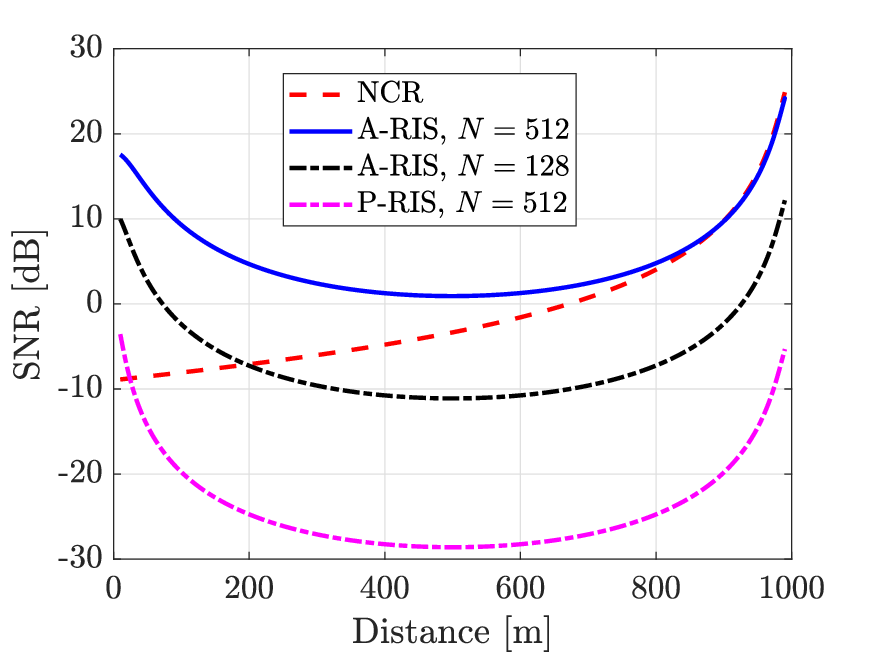}
\end{overpic} 
        \caption{The SNR obtained with passive/active RIS and NCR in terms of the horizontal position of the RIS/NCR. }
        \label{fig5}
\end{figure}

Up to this point, we have assumed that the BS is equipped with $M=1024$ antennas, which is consistent with the concept of gigantic MIMO operating in the upper midband \cite{bjornson2025enabling}. In Fig.~\ref{fig6}, we investigate the impact of the number of BS antennas $M$ on the performance gap between the NCR and passive/active RIS under the previously considered setup, with the UE now located at $(1000,0,0)$\,m and the RISs as well as the NCR positioned at $(800,10,10)$\,m.

As illustrated in Fig.~\ref{fig6}, the NCR consistently outperforms both passive and active RIS in terms of SNR when the number of BS antennas is less than 800. Notably, when the number of BS antennas is smaller, the relative performance gain provided by the NCR becomes more pronounced compared to both RIS-based solutions. On the other hand, active RIS with $N=512$ elements provides slightly larger SNR when having a gigantic number of BS antennas.

\begin{figure}[t!]
        \centering
	\begin{overpic}[width=0.98\columnwidth,trim=0.2cm 0cm 0.5cm 0.5cm,clip,tics=10]{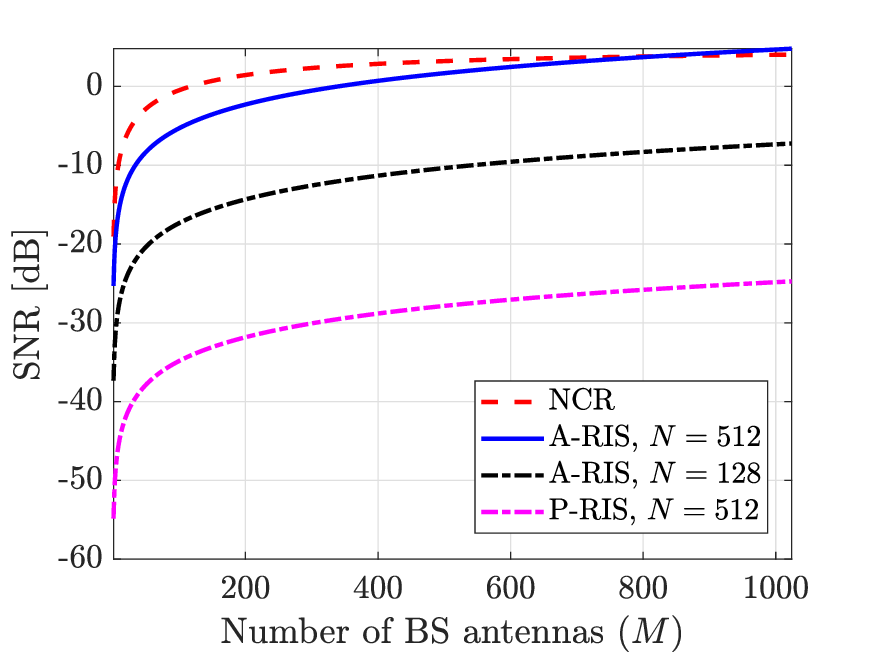}
\end{overpic} 
        \caption{The SNR obtained with passive/active RIS and NCR in terms of the number of BS antennas $M$. }
        \label{fig6}
\end{figure}

\section{Comparison When the Direct Path is Included} \label{sec4}

In this section, we will analyze the impact of a direct link between the BS and UE on the performance comparison of NCR and passive/active RIS. We let $\vect{h}_3$ denote the NLOS channel between the BS and the UE. Next, we will quantify the maximum SNR that can be achieved with a passive RIS, active RIS, and NCR.

\subsection{Passive RIS-assisted communication}
The overall RIS-assisted cascaded channel is 
\begin{align}
    \vect{h}= \vect{h}_3+\vect{H}_2\vect{\Psi}_{\rm P}\vect{h}_1. 
\end{align}
The optimal RIS phase-shift configuration is given by
\begin{align}
    \psi_{{\rm P}, n} = \exp\left(\imagunit(-\angle{[\vect{a}_1]_n}-\angle{[\vect{a}_{2,2}]_n}-\angle{\vect{h}_3^{\Htran}\vect{a}_{2,1}})\right),
\end{align}
which yields
\begin{align}
\left\lVert \vect{h} \right\rVert^2
&= \left\lVert \vect{h}_3 + \vect{H}_2 \vect{\Psi}_{\rm P} \vect{h}_1 \right\rVert^2 \nonumber \\
&= \left\lVert \vect{h}_3 \right\rVert^2 
+ \left\lVert \vect{H}_2 \vect{\Psi}_{\rm P} \vect{h}_1 \right\rVert^2 
+ 2 \Re\!\left( \vect{h}_3^{\Htran} \vect{H}_2 \vect{\Psi}_{\rm P} \vect{h}_1 \right) \nonumber \\
&= \left\lVert \vect{h}_3 \right\rVert^2 
+ \beta_1\beta_2 \left| \vect{a}_{2,2}^{\Ttran}\vect{\Psi}_{\rm P}\vect{a}_1 \right|^2 \left\lVert \vect{a}_{2,1} \right\rVert^2 \nonumber\\
&\quad + 2\sqrt{\beta_1\beta_2}\,\Re\!\left( \vect{h}_3^{\Htran}\vect{a}_{2,1}\vect{a}_{2,2}^{\Ttran}\vect{\Psi}_{\rm P}\vect{a}_1 \right) \nonumber \\
&= \left\lVert\vect{h}_3\right\rVert^2
+ \beta_1 \beta_2 N^2 M
+ 2\sqrt{\beta_1 \beta_2}N
\left|\vect{h}_3^{\Htran}\vect{a}_{2,1}\right|.
\end{align}
Thus, the maximum SNR in the passive RIS-assisted system becomes
\begin{align}
    \mathrm{SNR}_{\rm P-RIS} = \frac{P\left(\left\Vert\vect{h}_3\right\Vert^2+\beta_1\beta_2N^2M+2\sqrt{\beta_1\beta_2}N|\vect{h}_3^{\Htran}\vect{a}_{2,1}|\right)}{\sigma^2}.
\end{align}

\subsection{Active RIS-assisted communication}

 The received signal at the BS is given as 
\begin{align}
    \vect{y}_2 =  \sqrt{P}\vect{H}_2\vect{\Psi}_{\rm A} \vect{h}_1s+\sqrt{P}\vect{h}_3s+\vect{H}_2\vect{\Psi}_{\rm A}\vect{n}_1 + \vect{n}_2.
\end{align}
 The optimal receiver is the minimum mean-squared error (MMSE) receiver, which leads to the SNR in \eqref{eq:SNR-direct-activeRIS} at the top of the next page.

\begin{figure*}
\begin{align}
    \mathrm{SNR}_{\rm A-RIS} = P\left(\sqrt{\beta_1\beta_2}\vect{a}_{2,1}\vect{a}_{2,2}^{\Ttran}\vect{\Psi}_{\rm A} \vect{a}_1+\vect{h}_3\right)^{\Htran}\left(\beta_2\sigma^2\sum_{n=1}^N\left|\psi_{{\rm A},n}\right|^2\vect{a}_{2,1}\vect{a}_{2,1}^{\Htran}+\sigma^2\vect{I}_M\right)^{-1}\left(\sqrt{\beta_1\beta_2}\vect{a}_{2,1}\vect{a}_{2,2}^{\Ttran}\vect{\Psi}_{\rm A} \vect{a}_1+\vect{h}_3\right). \label{eq:SNR-direct-activeRIS}
\end{align}
\hrulefill
\end{figure*}

\subsubsection*{Simplification via Rank-One Update}
To simplify the SNR expression, we employ the rank-one update formula, which states that for an invertible matrix~$\mathbf{A}$ and a rank-one perturbation~$\mathbf{u}\mathbf{v}^{\Htran}$,
\begin{equation}
(\mathbf{A} + \mathbf{u} \mathbf{v}^{\Htran})^{-1} = \mathbf{A}^{-1} - \frac{\mathbf{A}^{-1} \mathbf{u} \mathbf{v}^{\Htran} \mathbf{A}^{-1}}{1 + \mathbf{v}^{\Htran} \mathbf{A}^{-1} \mathbf{u}}.
\end{equation}
Defining
\begin{align}
   & \mathbf{A} = \sigma^2 \mathbf{I}_M, \\
   & \mathbf{u} = \beta_2\sigma^2\sum_{n=1}^N\left|\psi_{{\rm A},n}\right|^2 \vect{a}_{2,1}, \\
   & \mathbf{v}^{\Htran} = \vect{a}_{2,1}^{\Htran},
\end{align}
and applying the formula to the inverse term in \eqref{eq:SNR-direct-activeRIS}, we obtain the expression in \eqref{eq:SNR-complicated-activeRIS} at the top of the next page.

\begin{figure*}
\begin{align}
    \mathrm{SNR}_{\rm A-RIS} &= \frac{PM\beta_1\beta_2\left|\vect{a}_{2,2}^{\Ttran}\vect{\Psi}_{\rm A}\vect{a}_1\right|^2-P\beta_2\sum_{n=1}^N\left|\psi_{{\rm A},n}\right|^2|\vect{h}_3^{\Htran}\vect{a}_{2,1}|^2+2P\sqrt{\beta_1\beta_2}\Re(\vect{h}_3^{\Htran}\vect{a}_{2,1}\vect{a}_{2,2}^{\Ttran}\vect{\Psi}_{\rm A}\vect{a}_1)}{\sigma^2(1+\sum_{n=1}^N\left|\psi_{{\rm A},n}\right|^2\beta_2M)}+\frac{P\Vert \vect{h}_3\Vert^2}{\sigma^2}. \label{eq:SNR-complicated-activeRIS}
\end{align}
\hrulefill
\end{figure*}

 The RIS phase-shift configuration that maximizes the SNR can easily be shown as
\begin{align}
    \psi_{{\rm A},n} =\vert \psi_{{\rm A},n}\vert\exp\left(\imagunit(-\angle{[\vect{a}_1]_n}-\angle{[\vect{a}_{2,2}]_n}-\angle{\vect{h}_3^{\Htran}\vect{a}_{2,1}})\right),
\end{align}
which yields the SNR in terms of RIS amplitude gains in \eqref{eq:SNR-complicated-activeRIS2} at the top of the next page.
\begin{figure*}
\begin{align}
    \mathrm{SNR}_{\rm A-RIS} &= \frac{PM\beta_1\beta_2\left(\sum_{n=1}^N\left|\psi_{A,n}\right|\right)^2-P\beta_2\sum_{n=1}^N\left|\psi_{{\rm A},n}\right|^2|\vect{h}_3^{\Htran}\vect{a}_{2,1}|^2+2P\sqrt{\beta_1\beta_2}\left\vert\vect{h}_3^{\Htran}\vect{a}_{2,1}\right\vert\sum_{n=1}^N|\psi_{{\rm A},n}|}{\sigma^2(1+\sum_{n=1}^N\left|\psi_{{\rm A},n}\right|^2\beta_2M)}+\frac{P\Vert \vect{h}_3\Vert^2}{\sigma^2}. \label{eq:SNR-complicated-activeRIS2}
\end{align}
\hrulefill
\end{figure*}

\begin{theorem}
The SNR in \eqref{eq:SNR-complicated-activeRIS2} is maximized when all RIS amplitude gains are identical, i.e., $|\psi_{{\rm A},n}| = \alpha_{\rm A\text{-}RIS}$ for all $n$, under the maximum amplitude constraint $\alpha_{\rm A-RIS}\leq \alpha_{\rm A-RIS, max}$. 

Under this configuration, the SNR can be expressed as
\begin{align}
\mathrm{SNR}_{\rm A\text{-}RIS}
= \frac{\mathcal{A}\alpha_{\rm A\text{-}RIS}^2+\mathcal{B}\alpha_{\rm A\text{-}RIS}}{\sigma^2+\mathcal{C}\alpha_{\rm A\text{-}RIS}^2}
+\frac{P\Vert \vect{h}_3\Vert^2}{\sigma^2},
\end{align}
where
\begin{align}
\mathcal{A}&= PM\beta_1\beta_2N^2-P\beta_2N|\vect{h}_3^{\Htran}\vect{a}_{2,1}|^2,\\
\mathcal{B}&=2PN\sqrt{\beta_1\beta_2}|\vect{h}_3^{\Htran}\vect{a}_{2,1}|,\\
\mathcal{C}&=\sigma^2\beta_2MN.
\end{align}

Moreover, the optimal amplitude is given by
\begin{align}
\alpha_{\rm A\text{-}RIS}
=
\begin{cases}
\alpha_{\rm A\text{-}RIS,max}, & \text{if } |\vect{h}_3^{\Htran}\vect{a}_{2,1}|=0,\\
\min\left(\alpha_{\rm A\text{-}RIS}^{\star}, \alpha_{\rm A\text{-}RIS,max}\right), & \text{otherwise},
\end{cases}
\end{align}
where
\begin{align}
\alpha_{\rm A\text{-}RIS}^{\star}
=\frac{\mathcal{A}\sigma^2+\sqrt{\mathcal{A}^2\sigma^4+\mathcal{B}^2\mathcal{C}\sigma^2}}{\mathcal{B}\mathcal{C}}.
\end{align}
\end{theorem}
\begin{proof}
    The proof is given in Appendix~A.
\end{proof}

The result shows that, although increasing the amplification improves the coherent signal gain, it also amplifies noise. When a direct link is present, excessive amplification can therefore be detrimental, which explains why the optimal gain is not always the maximum value.

In comparing the active RIS with NCR, we will use the optimal amplitude configuration derived above.

\subsection{NCR-assisted communication}

For the repeater-assisted system, the received signal at the BS is
\begin{align}
    y_2 &= \alpha_{\rm NCR} y_1 \vect{h}_2 +\sqrt{P}\vect{h}_3s +\vect{n}_2 \nonumber\\
        &= \alpha_{\rm NCR}\sqrt{P\beta_1\beta_2}\,\vect{a}_{2,1}e^{\imagunit \varphi_1}s \nonumber\\
        & \quad
        + \alpha_{\rm NCR}\sqrt{\beta_2}\,n_1\vect{a}_{2,1} 
        + \sqrt{P}\,\vect{h}_3s+\vect{n}_2,
\end{align}
where $n_1 \sim \CN(0,\sigma^2)$ and $\vect{n}_2 \sim \CN(\vect{0}, \sigma^2 \vect{I}_M)$.  
The optimal receiver is the MMSE receiver, which leads to the SNR in \eqref{eq:SNR-direct-repeater} at the top of the next page.

\begin{figure*}
\begin{align}
    \mathrm{SNR}_{\rm NCR} = P\left(\alpha_{\rm NCR}\sqrt{\beta_1\beta_2} \vect{a}_{2,1}e^{\imagunit \varphi_1}+\vect{h}_3\right)^{\Htran}\left(\alpha_{\rm NCR}^2\beta_2\sigma^2\vect{a}_{2,1}\vect{a}_{2,1}^{\Htran}+\sigma^2\vect{I}_M\right)^{-1}\left(\alpha_{\rm NCR}\sqrt{\beta_1\beta_2} \vect{a}_{2,1}e^{\imagunit \varphi_1}+\vect{h}_3\right). \label{eq:SNR-direct-repeater}
\end{align}
\hrulefill
\end{figure*}

Defining
\begin{align}
   & \mathbf{A} = \sigma^2 \mathbf{I}_M, \\
   & \mathbf{u} = \alpha_{\rm NCR}^2 \beta_2 \sigma^2 \vect{a}_{2,1}, \\
   & \mathbf{v}^{\Htran} = \vect{a}_{2,1}^{\Htran},
\end{align}
and applying the rank one update formula to the inverse term in \eqref{eq:SNR-direct-repeater}, we obtain the expression in \eqref{eq:SNR-complicated}.

\begin{figure*}
\begin{align}
    \mathrm{SNR}_{\rm NCR} &= \frac{P\alpha_{\rm NCR}^2M\beta_1\beta_2-P\alpha_{\rm NCR}^2\beta_2|\vect{h}_3^{\Htran}\vect{a}_{2,1}|^2+2P\alpha_{\rm NCR}\sqrt{\beta_1\beta_2}\Re(\vect{h}_3^{\Htran}\vect{a}_{2,1}e^{\imagunit \varphi_1})}{\sigma^2(1+\alpha_{\rm NCR}^2\beta_2M)}+\frac{P\Vert \vect{h}_3\Vert^2}{\sigma^2}. \label{eq:SNR-complicated}
\end{align}
\hrulefill
\end{figure*}

To find the optimal $\alpha_{\rm NCR}$, we start by defining 
\begin{align}
  &  \mathcal{A}= PM\beta_1\beta_2-P\beta_2|\vect{h}_3^{\Htran}\vect{a}_{2,1}|^2,\\
  &\mathcal{B} = 2P\sqrt{\beta_1\beta_2}\Re(\vect{h}_3^{\Htran}\vect{a}_{2,1}e^{\imagunit \varphi_1}), \\
  &\mathcal{C}=\sigma^2\beta_2M,
\end{align}
the SNR can be written in terms of $\alpha_{\rm NCR}$ as
\begin{align}
    \mathrm{SNR}_{\rm NCR}=\underbrace{\frac{\mathcal{A}\alpha_{\rm NCR}^2+\mathcal{B}\alpha_{\rm NCR}}{\sigma^2+\mathcal{C}\alpha_{\rm NCR}^2}}_{\triangleq f(\alpha_{\rm NCR})}+\frac{P\Vert \vect{h}_3\Vert^2}{\sigma^2}.
\end{align}
Hence, to maximize the SNR, we need to maximize the function $ f(\alpha_{\rm NCR})$ under the maximum amplitude constraint $\alpha_{\rm NCR}\leq \alpha_{\rm NCR, max}$. There are three cases to be evaluated:

\emph{Case 1:} $\vect{h}_3^{\Htran}\vect{a}_{2,1}=0$. In this case, we have $\mathcal{B}=0$ and $\mathcal{A}>0$ and the function $ f(\alpha_{\rm NCR})$ becomes monotonically increasing with $\alpha_{\rm NCR}$. Hence, the optimal amplification gain is $\alpha_{\rm NCR}=\alpha_{\rm NCR, max}$.

\emph{Case 2:} $\Re(\vect{h}_3^{\Htran}\vect{a}_{2,1}e^{\imagunit\psi_1})<0$. In this case, by equating the first derivative of $ f(\alpha_{\rm NCR})$ to zero and by the second derivative check, we obtain a minimizer. Hence, the solution is one of the boundary points, i.e., $0$ or $\alpha_{\rm NCR,max}$. By checking the objective value, the solution in this case is given as

\begin{align}
    \alpha_{\rm NCR}=\begin{cases} \alpha_{\rm NCR,max} & \text{if }\mathcal{A}\alpha_{\rm NCR,max}^2+\mathcal{B}\alpha_{\rm NCR,max}>0, \\
    0, & \text{otherwise}.\end{cases}
\end{align}
 
\emph{Case 3:} $\Re(\vect{h}_3^{\Htran}\vect{a}_{2,1}e^{\imagunit\psi_1})>0$. In this case, we have $\mathcal{B}>0$ and the derivation is the same as the one in the active RIS case. Equating the first derivative to zero, we obtain the optimal solution as
\begin{align}
    \alpha_{\rm NCR} = \min\left(\alpha_{\rm NCR}^{\star}, \alpha_{\rm NCR,max} \right),
\end{align}
where
\begin{align}
    \alpha_{\rm NCR}^{\star}=\frac{\mathcal{A}\sigma^2+\sqrt{\mathcal{A}^2\sigma^4+\mathcal{B}^2\mathcal{C}\sigma^2}}{\mathcal{B}\mathcal{C}}.
\end{align}

\begin{remark} While the overall derivation shares similarities with the active RIS case, the NCR optimization problem is inherently different due to the presence of the real-part term. This term couples amplitude and phase effects and prevents a direct reduction to a purely amplitude-dependent formulation. Consequently, the structure of the optimal solution differs and requires a separate case-by-case analysis. \end{remark}

\subsection{Comparison between passive RIS and NCR}
For the repeater to outperform the passive RIS, the SNR of the repeater-assisted link must exceed that of the RIS-assisted link, which is stated as a condition in \eqref{eq:condition} at the top of the next page.
\begin{figure*}
\begin{align}
   & \frac{P\alpha_{\rm NCR}^2M\beta_1\beta_2-P\alpha_{\rm NCR}^2\beta_2|\vect{h}_3^{\Htran}\vect{a}_{2,1}|^2+2P\alpha_{\rm NCR}\sqrt{\beta_1\beta_2}\Re(\vect{h}_3^{\Htran}\vect{a}_{2,1}e^{\imagunit \varphi_1})}{\sigma^2(1+\alpha_{\rm NCR}^2\beta_2M)}+\frac{P\Vert \vect{h}_3\Vert^2}{\sigma^2} \nonumber\\
    &\geq  \frac{P\left(\left\Vert\vect{h}_3\right\Vert^2+\beta_1\beta_2N^2M+2\sqrt{\beta_1\beta_2}N|\vect{h}_3^{\Htran}\vect{a}_{2,1}|\right)}{\sigma^2}. \label{eq:condition}
\end{align}
\hrulefill
\end{figure*}
Canceling $P\Vert \vect{h}_3\Vert^2/\sigma^2$ and multiplying by $\sigma^2/P$, we obtain
\begin{align}
f(\alpha_{\rm NCR},\vect{h}_3^{\Htran}\vect{a}_{2,1}e^{\imagunit \varphi_1})\geq g(N,\vect{h}_3^{\Htran}\vect{a}_{2,1}e^{\imagunit \varphi_1}), \label{eq:constraint-direct}
\end{align}
where
\begin{align}
&f(\alpha_{\rm NCR},x) \nonumber\\
&= \frac{\alpha_{\rm NCR}^2M\beta_1\beta_2-\alpha_{\rm NCR}^2\beta_2|x|^2+2\alpha_{\rm NCR}\sqrt{\beta_1\beta_2}\Re(x)}{1+\alpha_{\rm NCR}^2\beta_2M}, \label{eq:function-f} \\
&g(N,x) = \beta_1\beta_2 N^2M + 2\sqrt{\beta_1\beta_2}N|x|,
\end{align}
with $x\triangleq \vect{h}_3^{\Htran}\vect{a}_{2,1}e^{\imagunit \varphi_1}$.  

This condition reduces to a quadratic inequality
\begin{align}
\mathcal{A}\alpha_{\rm NCR}^2+\mathcal{B}\alpha_{\rm NCR}+\mathcal{C}\geq 0,
\end{align}
where
\begin{align}
\mathcal{A}&= M\beta_1\beta_2-|x|^2\beta_2-\beta_2Mg(N,x), \label{eq:A} \\
\mathcal{B} &=2\Re(x)\sqrt{\beta_1\beta_2}, \label{eq:B} \\
\mathcal{C} &=-g(N,x). \label{eq:C}
\end{align}
Since $\mathcal{C}<0$, three cases arise:  

\emph{Case 1:} $\mathcal{A}>0$.  
If $\alpha_{\rm NCR}> \frac{-\mathcal{B}+\sqrt{\mathcal{B}^2-4\mathcal{A}\mathcal{C}}}{2\mathcal{A}}$, the repeater outperforms the passive RIS.  

\emph{Case 2:} $\mathcal{A}=0$.  
If $\Re(x)\leq 0$, the repeater cannot beat the RIS regardless of the amplification gain. If $\Re(x)>0$ and $\alpha_{\rm NCR} > -\mathcal{C}/\mathcal{B}$, the repeater is superior.  

\emph{Case 3:} $\mathcal{A}<0$.  
If $\Re(x)\leq 0$, the repeater cannot outperform the RIS. If $\Re(x)>0$ and $\mathcal{B}^2-4\mathcal{A}\mathcal{C}\geq 0$, then the repeater outperforms the passive RIS whenever
\begin{align}
\alpha_{\rm NCR} \in \left( \frac{-\mathcal{B}+\sqrt{\mathcal{B}^2-4\mathcal{A}\mathcal{C}}}{2\mathcal{A}},\frac{-\mathcal{B}-\sqrt{\mathcal{B}^2-4\mathcal{A}\mathcal{C}}}{2\mathcal{A}}\right).
\end{align}

\subsubsection*{Numerical Illustration}
To illustrate, we consider free-space path loss with $f_c=15$\,GHz and $N=512$ RIS elements for the BS--RIS/NCR and RIS/NCR--UE channels. The BS--UE distance is $d_3=100$\,m, while the BS--RIS/NCR and RIS/NCR--UE distances are $d_2=90$\,m and $d_1=20$\,m, respectively.  The path loss for the direct channel, denoted by $\beta_3$, is given in decibel scale according to Urban Microcell Street Canyon model as $-32.4-20\log_{10}(15)-31.9\log_{10}(d_3/(1\,\text{m}))$ for $f_c=15$\,GHz \cite[Table 7.4.1-1]{3GPP5G}. The number of BS antennas is $M=1024$.
In Fig.~\ref{fig7}, the $x$-axis represents the angular mismatch between the direct path $\vect{h}_3$ and the cascaded path $\vect{a}_{2,1}e^{\imagunit\varphi_1}$.  

As the figure illustrates, when $x=\vect{h}_3^{\Htran}\vect{a}_{2,1}e^{\imagunit \varphi_1}$ is small in magnitude, the repeater can outperform the RIS with a small amplification. However, as $|x|$ grows, more amplification is needed at the NCR to outperform the passive RIS. Moreover, as the angular mismatch increases, $\mathcal{B}$ in \eqref{eq:B} eventually becomes negative and making the required $\alpha_{\rm NCR}$ higher as the previous analysis demonstrated.

\begin{figure}[t!]
        \centering
	\begin{overpic}[width=0.98\columnwidth,trim=0.2cm 0cm 0.5cm 0.5cm,clip,tics=10]{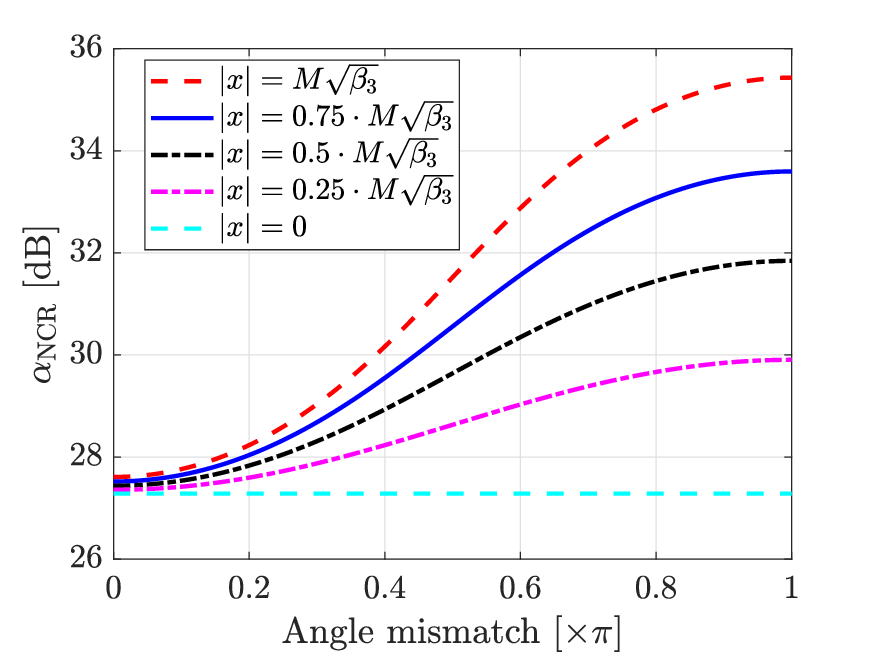}
\end{overpic} 
        \caption{The required amplification gain $\alpha_{\rm NCR}$ of the repeater to outperform the passive RIS in terms of the angle mismatch between the direct and cascaded path. }
        \label{fig7}
\end{figure}

 \subsection{Comparison between active RIS and NCR}
For the repeater to outperform the active RIS, the SNR of the repeater-assisted link must exceed that of the RIS-assisted link, i.e., as given in the condition in \eqref{eq:condition2} at the top of the next page.
\begin{figure*}
\begin{align}
   & \frac{P\alpha_{\rm NCR}^2M\beta_1\beta_2-P\alpha_{\rm NCR}^2\beta_2|\vect{h}_3^{\Htran}\vect{a}_{2,1}|^2+2P\alpha_{\rm NCR}\sqrt{\beta_1\beta_2}\Re(\vect{h}_3^{\Htran}\vect{a}_{2,1}e^{\imagunit \varphi_1})}{\sigma^2(1+\alpha_{\rm NCR}^2\beta_2M)}+\frac{P\Vert \vect{h}_3\Vert^2}{\sigma^2} \nonumber\\
    &\geq  \frac{P\alpha_{\rm A-RIS}^2M\beta_1\beta_2N^2-P\alpha_{\rm A-RIS}^2\beta_2N|\vect{h}_3^{\Htran}\vect{a}_{2,1}|^2+2P\alpha_{\rm A-RIS}N\sqrt{\beta_1\beta_2}\left\vert\vect{h}_3^{\Htran}\vect{a}_{2,1}\right\vert}{\sigma^2(1+\alpha_{\rm A-RIS}^2\beta_2MN)}+\frac{P\Vert \vect{h}_3\Vert^2}{\sigma^2}. \label{eq:condition2}
\end{align}
\hrulefill
\end{figure*}

Canceling $\frac{P\Vert \vect{h}_3\Vert^2}{\sigma^2}$ and multiplying by $\sigma^2/P$, we obtain
\begin{align}
f(\alpha_{\rm NCR},\vect{h}_3^{\Htran}\vect{a}_{2,1}e^{\imagunit \varphi_1})\geq g(N,\alpha_{\rm A-RIS},\vect{h}_3^{\Htran}\vect{a}_{2,1}e^{\imagunit \varphi_1}), \label{eq:constraint-direct}
\end{align}
where $f(\alpha_{\rm NCR},x)$ is given in \eqref{eq:function-f} and $g(N,\alpha_{\rm A-RIS},x)$ is given in \eqref{eq:function-g} at the top of the next page
\begin{figure*}
\begin{align}
&g(N,\alpha_{\rm A-RIS},x) = \frac{\alpha_{\rm A-RIS}^2M\beta_1\beta_2N^2-\alpha_{\rm A-RIS}^2\beta_2N|x|^2+2\alpha_{\rm A-RIS}N\sqrt{\beta_1\beta_2}\left\vert x \right\vert}{1+\alpha_{\rm A-RIS}^2\beta_2MN}. \label{eq:function-g}
\end{align}
\hrulefill
\end{figure*}
with $x\triangleq \vect{h}_3^{\Htran}\vect{a}_{2,1}e^{\imagunit \varphi_1}$.  

Similar to the comparison between the passive RIS, this condition reduces to a quadratic inequality 
\begin{align}
\mathcal{A}\alpha_{\rm NCR}^2+\mathcal{B}\alpha_{\rm NCR}+\mathcal{C}\geq 0,
\end{align}
where
\begin{align}
\mathcal{A}&= M\beta_1\beta_2-|x|^2\beta_2-\beta_2Mg(N,\alpha_{\m A-RIS},x), \label{eq:A2} \\
\mathcal{B} &=2\Re(x)\sqrt{\beta_1\beta_2}, \label{eq:B2} \\
\mathcal{C} &=-g(N,\alpha_{\rm A-RIS},x). \label{eq:C2}
\end{align}
Since $\mathcal{C}<0$, three cases arise:  

\emph{Case 1:} $\mathcal{A}>0$.  
If $\alpha_{\rm NCR}> \frac{-\mathcal{B}+\sqrt{\mathcal{B}^2-4\mathcal{A}\mathcal{C}}}{2\mathcal{A}}$, the repeater outperforms the active RIS.  

\emph{Case 2:} $\mathcal{A}=0$.  
If $\Re(x)\leq 0$, the repeater cannot beat the RIS. If $\Re(x)>0$ and $\alpha_{\rm NCR} > -\mathcal{C}/\mathcal{B}$, the repeater is superior.  

\emph{Case 3:} $\mathcal{A}<0$.  
If $\Re(x)\leq 0$, the repeater cannot outperform the RIS. If $\Re(x)>0$ and $\mathcal{B}^2-4\mathcal{A}\mathcal{C}\geq 0$, then the repeater outperforms the active RIS whenever
\begin{align}
\alpha_{\rm NCR} \in \left( \frac{-\mathcal{B}+\sqrt{\mathcal{B}^2-4\mathcal{A}\mathcal{C}}}{2\mathcal{A}},\frac{-\mathcal{B}-\sqrt{\mathcal{B}^2-4\mathcal{A}\mathcal{C}}}{2\mathcal{A}}\right).
\end{align}

\subsubsection*{Numerical Illustration}
We consider the same setup as in Fig.~\ref{fig7} and repeat the experiment using an active RIS with $N=512$ elements and the optimal amplitude configuration, where the maximum per-element amplification gain is $\alpha_{\rm A\text{-}RIS,max}=30$. 
When the BS--UE distance is set to $d_3=100$\,m, and the BS--RIS/NCR and RIS/NCR--UE distances are $d_2=90$\,m and $d_1=20$\,m, respectively, the repeater cannot outperform the active RIS for any amplification gain at such short distances. 
As observed in the previous sections, the main performance advantage of the NCR manifests itself in large-distance scenarios. 

Therefore, we repeat the same experiment by setting $d_1=1000$\,m, $d_2=900$\,m, and $d_3=200$\,m. 
Fig.~\ref{fig8} illustrates the required amplification gain for this setup, showing that significantly higher amplification levels are needed to outperform the active RIS. In particular, when the angular mismatch is large, extremely high amplification gains may be required, exceeding the practically reported limit of $90/2=45$\,dB \cite{bai2025repeater}.

\begin{figure}[t!]
        \centering
	\begin{overpic}[width=0.98\columnwidth,trim=0.2cm 0cm 0.5cm 0.5cm,clip,tics=10]{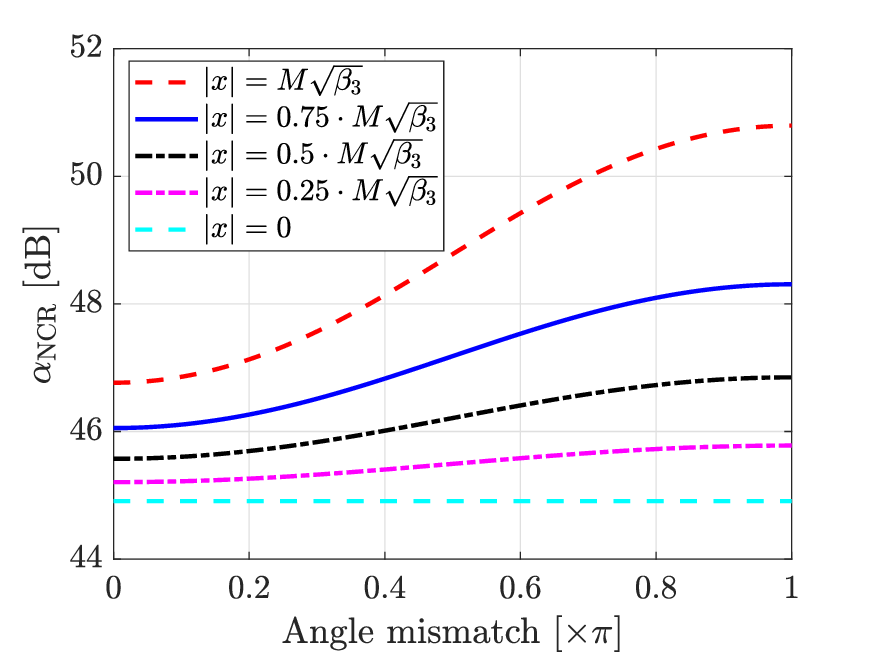}
\end{overpic} 
        \caption{The required amplification gain $\alpha_{\rm NCR}$ of the repeater to outperform the active RIS in terms of the angle mismatch between the direct and cascaded path. }
        \label{fig8}
\end{figure}

\section{Wideband Case}\label{sec5}

In the wideband case, all subcarriers are affected by the same amplification of NCR/active RIS or the same RIS phase-shifts. We consider a frequency-selective channel model where the direct and cascaded paths have different delay characteristics. 

The cascaded UE–RIS–BS (or UE–NCR–BS) link is assumed to be dominated by a LOS component and is therefore modeled as a single-tap channel. In contrast, the direct UE–BS link is modeled as a frequency-selective channel with multiple taps, capturing NLOS propagation effects.

This difference in delay spread leads to frequency-dependent variations across subcarriers, even if the cascaded channel itself is LOS:
\begin{align}
x=\vect{h}_3^{\Htran}\vect{a}_{2,1}e^{\imagunit \varphi_1}.
\end{align}
The $\vect{h}_3$ and $\varphi_1$ will vary across the subcarriers. Hence, even if the RIS phase-shifts are optimized for one subcarrier, phase misalignments inevitably occur on the other subcarriers. To account for frequency selectivity, we model the direct channel $\vect{h}_3$ as varying across subcarriers and consider the average SNR with respect to its distribution. Specifically, we assume that $\vect{h}_3 \sim \mathcal{CN}(\vect{0}, \vect{R}_3)$, where $\vect{R}_3$ is the spatial correlation matrix. Then, we have $\mathbb{E}\{\Vert\vect{h}_3\Vert^2\}=\tr(\vect{R}_3)=M\beta_3$, where $\beta_3$ is the corresponding path loss and shadowing-based large-scale channel fading coefficient.

\subsection{Passive RIS-assisted communication}
The average SNR obtained with the passive RIS is given by
\begin{align}
&\overline{\mathrm{SNR}}_{\rm P-RIS} \nonumber\\
&=
\frac{P}{\sigma^2}\Big(
\mathbb{E}\{\|\vect{h}_3\|^2\}
+\beta_1\beta_2N^2M
+2\sqrt{\beta_1\beta_2}N\,\mathbb{E}\big\{\Re(\vect{h}_3^{\Htran}\vect{a}_{2,1}e^{\imagunit\varphi_1})\big\}
\Big).
\end{align}
Since $\vect{h}_3 \sim \mathcal{CN}(\vect{0},\vect{R}_3)$ has zero mean, we obtain
\begin{align}
\mathbb{E}\big\{\Re(\vect{h}_3^{\Htran}\vect{a}_{2,1}e^{\imagunit\varphi_1})\big\}=0,
\end{align}
which leads to
\begin{align}
\overline{\mathrm{SNR}}_{\rm P-RIS}
=
\frac{P\left(\beta_3M+\beta_1\beta_2N^2M\right)}{\sigma^2}.
\end{align}

\subsection{Active RIS-assisted communication}
Similarly, averaging over random realizations of NLOS direct link at different subcarriers, the active RIS-assisted system yields
\begin{align}
   \overline{\mathrm{SNR}}_{\rm A-RIS}&=\frac{P\alpha_{\rm A-RIS}^2M\beta_1\beta_2N^2-P\alpha_{\rm A-RIS}^2\beta_2N\vect{a}_{2,1}^{\Htran}\vect{R}_3\vect{a}_{2,1}}{\sigma^2(1+\alpha_{\rm A-RIS}^2\beta_2MN)} \nonumber\\
   &\quad +\frac{P\beta_3M}{\sigma^2}.
\end{align}

\subsection{Repeater-assisted communication}
The average SNR for the repeater-assisted system is given as
\begin{align}
   \overline{\mathrm{SNR}}_{\rm NCR}&= \frac{P\alpha_{\rm NCR}^2M\beta_1\beta_2-P\alpha_{\rm NCR}^2\beta_2\vect{a}_{2,1}^{\Htran}\vect{R}_3\vect{a}_{2,1}}{\sigma^2(1+\alpha_{\rm NCR}^2\beta_2M)} \nonumber\\
   &\quad+\frac{P\beta_3M}{\sigma^2}.
\end{align}
Following steps similar to the narrowband case, one can derive analogous conditions for the repeater to outperform the passive/active RIS.  

An interesting operating regime arises when $\vect{a}_{2,1}$ lies in the null space of $\vect{R}_3$, i.e.,
$\vect{a}_{2,1}^{\Htran}\vect{R}_3\vect{a}_{2,1}=0$.
In this case, the conditions under which the NCR outperforms the passive or active RIS reduce to those obtained in the absence of a direct link, and coincide with the expressions derived in Section~\ref{sec3}.

In Fig.~\ref{fig9}, we consider a three-dimensional deployment scenario similar to that in Fig.~\ref{fig5}. The BS, equipped with $M=1024$ antennas, is located at $(0,0,10)$\,m, while the UE is positioned at $(1000,0,0)$\,m. The RIS or repeater is placed at $(r_x,10,10)$\,m, where $r_x$ denotes its horizontal location. We evaluate the channel gains as functions of $r_x$. The active RIS employs a fixed amplification factor of $\alpha_{\rm A\text{-}RIS}=30$. For the repeater, both a maximum amplification constraint of $45$\,dB and a total transmit power constraint of $1$\,W are imposed. The spatial correlation matrix is modeled as the superposition of Dirac delta impulses corresponding to four scatterers located in the vicinity of the UE at $(1000,5,0)$, $(1000,-5,0)$, $(990,5,0)$, and $(990,-5,0)$\,m.

Fig.~\ref{fig9} illustrates the average SNR achieved by the NCR, passive RIS, and active RIS as a function of their deployment location, while the BS and UE positions remain fixed. The \emph{Direct} case serves as a baseline and represents the average SNR obtained when only the direct BS--UE link is present. Owing to the large separation between the BS and UE, this SNR is very low. Employing a passive RIS with $N=512$ elements significantly improves the average SNR, particularly when the RIS is deployed close to either the BS or the UE. Nevertheless, due to the double-cascaded path loss, the SNR achieved with the passive RIS remains below $0$\,dB for all considered locations.

In contrast, the active RIS attains substantially higher SNR values, especially when placed in close proximity to the UE. For the active RIS configuration with $N=512$ elements, the NCR slightly outperforms the active RIS when positioned near the UE. On the other hand, when $N=128$, the NCR achieves markedly higher SNR than the active RIS over most deployment locations. Overall, the trends observed in Fig.~\ref{fig9} closely resemble those in Fig.~\ref{fig5}, and thus the same conclusions can be drawn regarding the average SNR behavior in the wideband case.

Until now, for the active RIS, we have limited the maximum amplification gain based on the circuit parameters reported in~\cite{gavriilidis2025active}. While some prior works consider higher amplification levels, there is currently no clear consensus on the practically achievable range.

To further investigate the performance under a fair power comparison, we repeat the previous experiment by constraining both the NCR and the active RIS to have the same total radiated power, set to $1$\,W. As shown in Fig.~\ref{fig10}, the active RIS with either $N=128$ or $N=512$ elements significantly outperforms the NCR. This gain stems from the coherent beamforming capability of the active RIS, enabled by phase shifts that are constructively aligned with the channel. However, the relative performance strongly depends on the achievable amplification levels, which are ultimately determined by the underlying hardware technology. Hence, which architecture provides superior performance depends on the practically realizable amplification capabilities.

\begin{figure}[t!]
        \centering
	\begin{overpic}[width=0.98\columnwidth,trim=0.2cm 0cm 0.5cm 0.5cm,clip,tics=10]{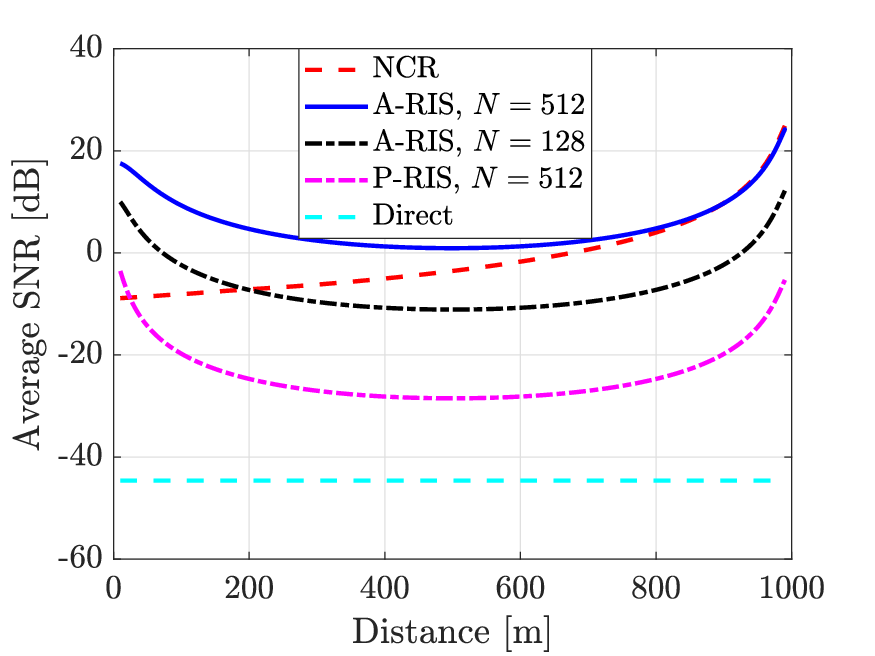}
\end{overpic} 
        \caption{The average SNR obtained with passive/active RIS and NCR in terms of the horizontal position of the RIS/NCR. The average SNR without any repeater or RIS is also shown as a baseline. }
        \label{fig9}
\end{figure}

\begin{figure}[t!]
        \centering
	\begin{overpic}[width=0.98\columnwidth,trim=0.2cm 0cm 0.5cm 0.5cm,clip,tics=10]{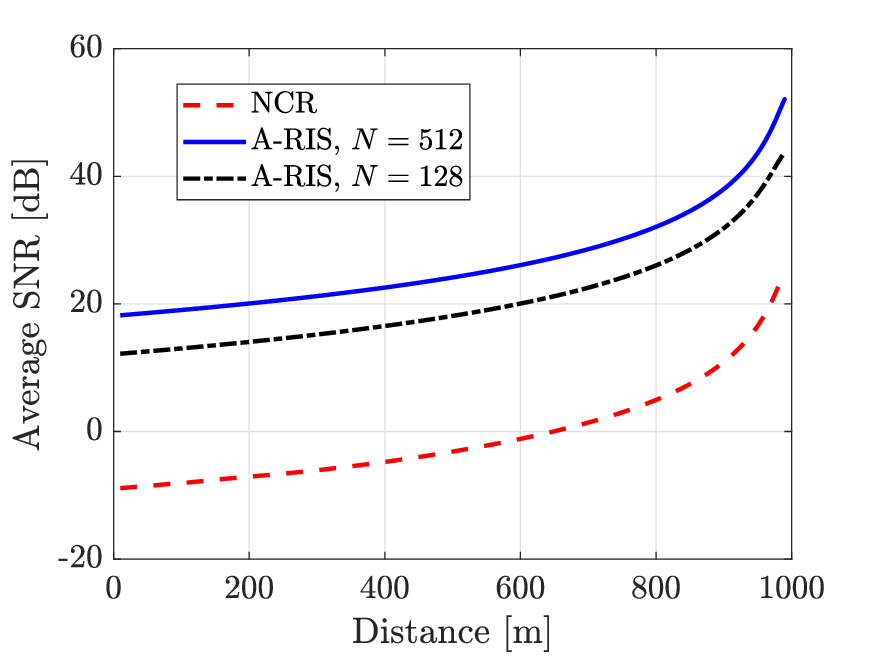}
\end{overpic} 
        \caption{The average SNR obtained with active RIS and NCR in terms of the horizontal position of the active RIS/NCR. The total radiated power from active RIS/NCR is $1$\,W. }
        \label{fig10}
\end{figure}

\section{Conclusions} \label{sec6}

This paper presented a comprehensive comparison between RISs and NCRs in terms of achievable SNR. By deriving analytical SNR expressions, we identified explicit conditions under which an NCR can outperform passive and active RISs. 

In narrowband systems without a direct BS--UE link, the condition for an NCR to outperform an RIS reduces to a required amplification level (as an increasing function of RIS elements). When a direct path is present, the interaction between the direct and cascaded channels introduces more intricate performance conditions, where constructive or destructive interference plays a decisive role. In wideband systems, frequency-dependent phase variations degrade the coherent combining gain of RISs, making NCRs relatively more competitive. Notably, when the direct and BS--repeater channels are orthogonal, the performance condition for the NCR coincides with that of the narrowband case without a direct link.

Numerical results further demonstrated that NCRs can outperform both passive and active RISs when deployed close to the UE, where only moderate amplification levels are required. Conversely, when the RIS is active and sufficiently large, or placed near the BS or UE, RIS-based solutions consistently achieve superior performance. These results indicate that neither technology is universally optimal; instead, the choice between RISs and NCRs should be guided by deployment geometry, system bandwidth, hardware constraints, and performance requirements. Future work includes extending the analysis to multi-user scenarios and networks with multiple NCRs.

\appendices

\section{Proof of Theorem 1}

We first show that the SNR in \eqref{eq:SNR-complicated-activeRIS2} is maximized when all RIS amplitude gains are identical. Observe that the SNR depends on the amplitudes through the terms
\begin{equation}
\left(\sum_{n=1}^N |\psi_{{\rm A},n}|\right)^2
\quad\text{and}\quad
\sum_{n=1}^N |\psi_{{\rm A},n}|^2,
\end{equation}
where the former appears in the numerator and the latter appears in both the numerator and the denominator.

Assume, by contradiction, that the optimal solution contains non-equal amplitudes $\{|\psi_{{\rm A},n}|\}$. Consider an alternative feasible configuration in which all amplitudes are set equal while keeping the sum $\sum_{n=1}^N |\psi_{{\rm A},n}|$ unchanged. Since the function $\sum_{n=1}^N |\psi_{{\rm A},n}|^2$ is strictly convex, Jensen’s inequality implies that this equal-amplitude configuration strictly decreases $\sum_{n=1}^N |\psi_{{\rm A},n}|^2$.

Consequently, the dominant quadratic gain term in the numerator remains unchanged, whereas the penalty terms involving $\sum_{n=1}^N |\psi_{{\rm A},n}|^2$ in both the numerator and the denominator are strictly reduced, which yields a strictly larger SNR. This contradicts the assumed optimality of unequal amplitudes.

Therefore, the optimal RIS amplitude gains must be identical across all RIS elements. We let $\alpha_{\rm A-RIS}$ denote the identical amplitude for each RIS element. Further defining, 
\begin{align}
  &  \mathcal{A}= PM\beta_1\beta_2N^2-P\beta_2N|\vect{h}_3^{\Htran}\vect{a}_{2,1}|^2,\\
  &\mathcal{B} = 2PN\sqrt{\beta_1\beta_2}|\vect{h}_3^{\Htran}\vect{a}_{2,1}|, \\
  &\mathcal{C}=\sigma^2\beta_2MN,
\end{align}
the SNR can be written in terms of $\alpha_{\rm A-RIS}$ as
\begin{align}
    \mathrm{SNR}_{\rm A-RIS}=\underbrace{\frac{\mathcal{A}\alpha_{\rm A-RIS}^2+\mathcal{B}\alpha_{\rm A-RIS}}{\sigma^2+\mathcal{C}\alpha_{\rm A-RIS}^2}}_{\triangleq f(\alpha_{\rm A-RIS})}+\frac{P\Vert \vect{h}_3\Vert^2}{\sigma^2}.
\end{align}
Hence, to maximize the SNR we need to maximize the function $ f(\alpha_{\rm A-RIS})$ under the maximum amplitude constraint $\alpha_{\rm A-RIS}\leq \alpha_{\rm A-RIS, max}$. There are two cases to be evaluated:

\emph{Case 1:} $|\vect{h}_3^{\Htran}\vect{a}_{2,1}|=0$. In this case, we have $\mathcal{B}=0$ and $\mathcal{A}>0$ and the function $ f(\alpha_{\rm A-RIS})$ becomes monotonically increasing with $\alpha_{\rm A-RIS}$. Hence, the optimal amplification gain is $\alpha_{\rm A-RIS}=\alpha_{\rm A-RIS, max}$.

\emph{Case 2:} $|\vect{h}_3^{\Htran}\vect{a}_{2,1}|>0$. In this case, we first relax the maximum amplitude constraint, take the derivative of $ f(\alpha_{\rm A-RIS})$ with respect to $\alpha_{\rm A-RIS}$, and equate it to zero, i.e.,
\begin{align}
   f'(\alpha_{\rm A-RIS})= -\frac{\mathcal{B}\mathcal{C}\alpha_{\rm A-RIS}^2-2\mathcal{A}\sigma^2\alpha_{\rm A-RIS}-\mathcal{B}\sigma^2}{\left(\sigma^2+\mathcal{C}\alpha_{\rm A-RIS}^2\right)^2}=0. \label{eq:first-derivative}
\end{align}
Noting that $\mathcal{B}>0$ and $\mathcal{C}>0$, there is only one positive real root of the polynomial in the numerator of the above expression. This root is given as
\begin{align}
    \alpha_{\rm A-RIS}^{\star}=\frac{\mathcal{A}\sigma^2+\sqrt{\mathcal{A}^2\sigma^4+\mathcal{B}^2\mathcal{C}\sigma^2}}{\mathcal{B}\mathcal{C}},
\end{align}
Using the second derivative test at $\alpha_{\rm A-RIS}=\alpha_{\rm A-RIS}^{\star}$, we obtain
\begin{align}
    f''(\alpha_{\rm A-RIS}^{\star})& = \frac{4\mathcal{C}\alpha_{\rm A-RIS}^{\star} \left(\mathcal{B}\mathcal{C}\alpha_{\rm A-RIS}^2-2\mathcal{A}\sigma^2\alpha_{\rm A-RIS}-\mathcal{B}\sigma^2\right)}{\left(\sigma^2+\mathcal{C}(\alpha_{\rm A-RIS}^{\star})^2\right)^3} \nonumber\\
    &\quad - \frac{2\mathcal{B}\mathcal{C}\alpha_{\rm A-RIS}^{\star} - 2\mathcal{A}\sigma^2}{\left(\sigma^2+\mathcal{C}(\alpha_{\rm A-RIS}^{\star})^2\right)^2} \nonumber\\
    &=-\frac{2\sqrt{\mathcal{A}^2\sigma^4+\mathcal{B}^2\mathcal{C}\sigma^2}}{\left(\sigma^2+\mathcal{C}(\alpha_{\rm A-RIS}^{\star})^2\right)^2}<0.
\end{align}
Hence, $\alpha_{\rm A-RIS}^{\star}$ is the maximizing point. Considering also the maximum amplitude constraints, the optimal $\alpha_{\rm A-RIS}$ is obtained as
\begin{align}
    \alpha_{\rm A-RIS} = \min\left(\alpha_{\rm A-RIS}^{\star}, \alpha_{\rm A-RIS,max} \right).
\end{align}

\bibliographystyle{IEEEtran}
\bibliography{IEEEabrv,refs}

\begin{thebibliography}{10}
\providecommand{\url}[1]{#1}
\csname url@samestyle\endcsname
\providecommand{\newblock}{\relax}
\providecommand{\bibinfo}[2]{#2}
\providecommand{\BIBentrySTDinterwordspacing}{\spaceskip=0pt\relax}
\providecommand{\BIBentryALTinterwordstretchfactor}{4}
\providecommand{\BIBentryALTinterwordspacing}{\spaceskip=\fontdimen2\font plus
\BIBentryALTinterwordstretchfactor\fontdimen3\font minus
  \fontdimen4\font\relax}
\providecommand{\BIBforeignlanguage}[2]{{%
\expandafter\ifx\csname l@#1\endcsname\relax
\typeout{** WARNING: IEEEtran.bst: No hyphenation pattern has been}%
\typeout{** loaded for the language `#1'. Using the pattern for}%
\typeout{** the default language instead.}%
\else
\language=\csname l@#1\endcsname
\fi
#2}}
\providecommand{\BIBdecl}{\relax}
\BIBdecl

\bibitem{liu2021reconfigurable}
Y.~Liu, X.~Liu, X.~Mu, T.~Hou, J.~Xu, M.~Di~Renzo, and N.~Al-Dhahir,
  ``Reconfigurable intelligent surfaces: Principles and opportunities,''
  \emph{IEEE communications surveys \& tutorials}, vol.~23, no.~3, pp.
  1546--1577, 2021.

\bibitem{bjornson2021reconfigurable}
E.~Bj{\"o}rnson, {\"O}.~{\"O}zdogan, and E.~G. Larsson, ``Reconfigurable
  intelligent surfaces: Three myths and two critical questions,'' \emph{IEEE
  Communications Magazine}, vol.~58, no.~12, pp. 90--96, 2021.

\bibitem{zhang2022active}
Z.~Zhang, L.~Dai, X.~Chen, C.~Liu, F.~Yang, R.~Schober, and H.~V. Poor,
  ``Active {RIS} vs. passive {RIS}: Which will prevail in {6G}?'' \emph{IEEE
  Transactions on Communications}, vol.~71, no.~3, pp. 1707--1725, 2022.

\bibitem{zhi2022active}
K.~Zhi, C.~Pan, H.~Ren, K.~K. Chai, and M.~Elkashlan, ``Active {RIS} versus
  passive {RIS}: Which is superior with the same power budget?'' \emph{IEEE
  Communications Letters}, vol.~26, no.~5, pp. 1150--1154, 2022.

\bibitem{basar2024reconfigurable}
E.~Basar, G.~C. Alexandropoulos, Y.~Liu, Q.~Wu, S.~Jin, C.~Yuen, O.~A. Dobre,
  and R.~Schober, ``Reconfigurable intelligent surfaces for {6G}: Emerging
  hardware architectures, applications, and open challenges,'' \emph{IEEE
  Vehicular Technology Magazine}, 2024.

\bibitem{wen2024shaping}
C.-K. Wen, L.-S. Tsai, A.~Shojaeifard, P.-K. Liao, K.-K. Wong, and C.-B. Chae,
  ``Shaping a smarter electromagnetic landscape: {IAB, NCR, and RIS} in {5G}
  standard and future {6G},'' \emph{IEEE Communications Standards Magazine},
  vol.~8, no.~1, pp. 72--78, 2024.

\bibitem{tsai2010capacity}
L.-S. Tsai and D.-s. Shiu, ``Capacity scaling and coverage for repeater-aided
  {MIMO} systems in line-of-sight environments,'' \emph{IEEE transactions on
  wireless communications}, vol.~9, no.~5, pp. 1617--1627, 2010.

\bibitem{willhammar2025achieving}
S.~Willhammar, H.~Iimori, J.~Vieira, L.~Sundstr{\"o}m, F.~Tufvesson, and E.~G.
  Larsson, ``Achieving distributed {MIMO} performance with repeater-assisted
  cellular massive {MIMO},'' \emph{IEEE Communications Magazine}, vol.~63,
  no.~3, pp. 114--119, 2025.

\bibitem{topal2025fair}
O.~A. Topal, {\"O}.~T. Demir, E.~Bj{\"o}rnson, and C.~Cavdar, ``Fair and
  energy-efficient activation control mechanisms for repeater-assisted massive
  mimo,'' \emph{arXiv preprint arXiv:2504.03428}, 2025.

\bibitem{demir2026pa}
{\"O}.~T. Demir and E.~Bj\"ornson, ``Network-controlled repeaters under power
  amplifier non-linearities,'' in \emph{IEEE International Conference on
  Acoustics, Speech and Signal Processing (ICASSP)}.\hskip 1em plus 0.5em minus
  0.4em\relax Barcelona, Spain: IEEE, May 2026, pp. 1--5, accepted and to be
  presented.

\bibitem{evgur2026ofdm}
D.~Evg{\"u}r, O.~A. Topal, and {\"O}.~T. Demir, ``Capacity analysis of {OFDM}
  systems with a swarm of network-controlled repeaters,'' in \emph{IEEE
  International Conference on Acoustics, Speech and Signal Processing
  (ICASSP)}.\hskip 1em plus 0.5em minus 0.4em\relax Barcelona, Spain: IEEE, May
  2026, pp. 1--5, accepted and to be presented.

\bibitem{ayoubi2022network}
R.~A. Ayoubi, M.~Mizmizi, D.~Tagliaferri, D.~D. Donno, and U.~Spagnolini,
  ``Network-controlled repeaters vs. reconfigurable intelligent surfaces for 6g
  mmw coverage extension: A simulative comparison,'' in \emph{2023 21st
  Mediterranean Communication and Computer Networking Conference (MedComNet)},
  2023, pp. 196--202.

\bibitem{guo2022comparison}
H.~Guo, C.~Madapatha, B.~Makki, B.~Dortschy, L.~Bao, M.~{\AA}str{\"o}m, and
  T.~Svensson, ``A comparison between network-controlled repeaters and
  reconfigurable intelligent surfaces,'' \emph{arXiv preprint
  arXiv:2211.06974}, 2022.

\bibitem{sun2023performance}
Y.~Sun, B.~Duan, X.~Su, H.~Wang, Q.~Gu, J.~Jin, and Y.~Yuan, ``Performance
  analysis on reconfigurable intelligent surface and network-controlled
  repeater in {3GPP} release-18,'' \emph{Frontiers of Information Technology \&
  Electronic Engineering}, vol.~24, no.~12, pp. 1815--1828, 2023.

\bibitem{andersson2025repeater}
M.~Andersson, A.~Chowdhury, and E.~G. Larsson, ``Is repeater-assisted massive
  {MIMO} compatible with dynamic {TDD}?'' \emph{arXiv preprint
  arXiv:2510.20998}, 2025.

\bibitem{chowdhury2025performance}
A.~Chowdhury and E.~G. Larsson, ``On the performance of dual-antenna repeater
  assisted bi-static {MIMO ISAC},'' \emph{arXiv preprint arXiv:2511.17980},
  2025.

\bibitem{wang2024reconfigurable}
J.~Wang, W.~Tang, J.~C. Liang, L.~Zhang, J.~Y. Dai, X.~Li, S.~Jin, Q.~Cheng,
  and T.~J. Cui, ``Reconfigurable intelligent surface: Power consumption
  modeling and practical measurement validation,'' \emph{IEEE Transactions on
  Communications}, vol.~72, no.~9, pp. 5720--5734, 2024.

\bibitem{bjornson2020power}
E.~Bj{\"o}rnson and L.~Sanguinetti, ``Power scaling laws and near-field
  behaviors of massive mimo and intelligent reflecting surfaces,'' \emph{IEEE
  Open Journal of the Communications Society}, vol.~1, pp. 1306--1324, 2020.

\bibitem{gavriilidis2025active}
P.~Gavriilidis, D.~Mishra, B.~Smida, E.~Basar, C.~Yuen, and G.~C.
  Alexandropoulos, ``Active reconfigurable intelligent surfaces: Circuit
  modeling and reflection amplification optimization,'' \emph{arXiv preprint
  arXiv:2503.24093}, 2025.

\bibitem{bai2025repeater}
J.~Bai, A.~Chowdhury, A.~Hansson, and E.~G. Larsson, ``Repeater swarm-assisted
  cellular systems: Interaction stability and performance analysis,''
  \emph{arXiv preprint arXiv:2508.13593}, 2025.

\bibitem{bjornson2025enabling}
E.~Bj{\"o}rnson, F.~Kara, N.~Kolomvakis, A.~Kosasih, P.~Ramezani, and M.~B.
  Salman, ``Enabling {6G} performance in the upper mid-band by transitioning
  from massive to gigantic {MIMO},'' \emph{IEEE Open Journal of the
  Communications Society}, 2025.

\bibitem{3GPP5G}
\emph{{5G}; Study on channel model for frequencies from 0.5 to 100 GHz (Release
  14)}.\hskip 1em plus 0.5em minus 0.4em\relax {3GPP} {TR} 38.901, Jan. 2018.

\end{thebibliography}

\end{document}